\documentclass[a4paper,11pt]{article}%
\usepackage{latexsym,amssymb,amsmath,amsfonts,amsthm}
\usepackage{graphicx}
\usepackage{subfig}
\usepackage{xcolor}
\usepackage{placeins}
\usepackage{amsmath}
\usepackage{amsfonts}
\usepackage{amssymb}%
\providecommand{\U}[1]{\protect\rule{.1in}{.1in}}
\newcommand{\R}{{\mathbb R}}

\newcommand{\C}{{\mathbb C}}

\newcommand{\be}{\begin{eqnarray}}
\newcommand{\ben}{\begin{eqnarray*}}
\newcommand{\en}{\end{eqnarray}}
\newcommand{\enn}{\end{eqnarray*}}

\newtheorem{theorem}{Theorem}[section]
\newtheorem{lemma}[theorem]{Lemma}

\newtheorem{defn}[theorem]{Definition}
\newtheorem{remark}[theorem]{Remark}
\newtheorem{proposition}[theorem]{Proposition}

\definecolor{rot}{rgb}{0.000,0.000,0.000}
\newcommand{\tcr}{\textcolor{rot}}

\begin{document}
\renewcommand{\theequation}{\arabic{section}.\arabic{equation}}
\begin{titlepage}
\title{\bf Direct and inverse time-harmonic elastic scattering from point-like and extended obstacles}
\author{Guanghui Hu$^1$, Andrea Mantile$^2$, Mourad Sini $^3$ and Tao Yin$^4$}
\let\thefootnote\relax\footnotetext{
$^{1}$Beijing Computational Science Research Center, Beijing 100193, China. Email: hu@csrc.ac.cn \\
$^{2}$LMR, EA4535 URCA, F\'ed\'eration de Recherche ARC Math\'ematiques, FR 3399
CNRS. Email: andrea.mantile@univ-reims.fr\\
$^{3}$RICAM, Austrian Academy of Sciences, Altenbergerstr. 69, A-4040 Linz, Austria. Email: mourad.sini@oeaw.ac.at (Corresponding author)\\
$^{4}$Department of Computing \& Mathematical Sciences, California Institute of Technology, 1200 East California Blvd., CA 91125, United States. Email: taoyin89@caltech.edu
}
\date{}
\end{titlepage}

\maketitle
\begin{abstract}

This paper concerns the time-harmonic direct and inverse elastic scattering by
an extended rigid elastic body surrounded by a finite number of point-like
obstacles. We first justify the point-interaction model for the Lam\'{e}
operator within the singular perturbation approach. For a general family of
pointwise-supported singular perturbations, including anisotropic and
non-local interactions, we derive an explicit representation of the
scattered field. In the case of isotropic and local point-interactions, our
result is consistent with the ones previously obtained by Foldy's formal method
as well as by the renormalization technique.

In the case of multiple scattering with pointwise and extended obstacles, we
show that the scattered field consists of two parts: one is due to the
diffusion by the extended scatterer and the other one is a linear combination
of the interactions between the point-like obstacles and the interaction
between the point-like obstacles with the extended one.

As to the inverse problem, the factorization method by Kirsch is adapted to
recover simultaneously the shape of an extended elastic body and the location
of point-like scatterers in the case of isotropic and local interactions. The
inverse problems using only one type of elastic waves (i.e. pressure or shear
waves) are also investigated and numerical examples are present to confirm the
inversion schemes.


\vspace{.1in}

\end{abstract}


\section{Introduction}

\setcounter{equation}{0}

We deal with the elastic scattering of a time-harmonic plane wave from an
inhomogeneous isotropic medium in ${\mathbb{R}}^{n}$ ($n=2,3$) characterized
by the mass density function $\rho:=\rho(x)$ and the Lam\'{e} constants
$\lambda,\mu\in{\mathbb{R}}$ satisfying
\begin{equation}
\mu>0,\quad n\lambda+2\mu>0. \label{ConL}%
\end{equation}
It is supposed that the background medium is homogeneous, isotropic and that
the inhomogeneous medium occupies a bounded domain $\Omega$ with the Lipschitz
boundary $\partial\Omega$. In particular, $\Omega$ is allowed to contain a
finite number of disconnected components, but its exterior $\Omega
^{e}:={\mathbb{R}}^{n}\backslash\overline{\Omega}$ is always connected. For
simplicity we assume that $\rho\equiv1$ in $\Omega^{e}:={\mathbb{R}}%
^{n}\backslash\overline{\Omega}$. In linear elasticity, the elastic
displacement is then governed by the time-harmonic Navier equation
\begin{equation}
(\Delta^{\ast}+\omega^{2})\,u=0\quad\mbox{in}\quad\Omega^{e},\qquad
\Delta^{\ast}:=\mu\Delta+(\lambda+\mu)\,\mathrm{grad\,}\mathrm{div\,}%
\quad\label{Navier}%
\end{equation}
where $\omega>0$ denotes the angular frequency of incitation and
$u=u^{in}+u^{sc}$ is the sum of the incident and scattered fields. Since the
domain $\Omega^{e}$ is infinity in all directions $\hat{x}:=x/|x|\in
\mathbb{S}^{n-1}:=\{|\hat{x}|=1\}$, the scattered field $u^{sc}$ is required
to satisfy the outgoing Kupradze radiation conditions
\begin{equation}%
\begin{array}
[c]{ccccc}%
\lim_{r\rightarrow\infty}r^{(n-1)/2}(\frac{\partial u_{p}}{\partial r}%
-ik_{p}u_{p})=0\,, &  & \lim_{r\rightarrow\infty}r^{(n-1)/2}(\frac{\partial
u_{s}}{\partial r}-ik_{s}u_{s})=0\,, &  & r=|x|\,,
\end{array}
\label{Kupradze}%
\end{equation}
uniformly in all directions. Here,
\begin{equation}%
\begin{array}
[c]{ccc}%
k_{p}:=\omega/\sqrt{\lambda+2\mu}\,, &  & k_{s}:=\omega/\sqrt{\mu}\,,
\end{array}
\label{wn}%
\end{equation}
are the compressional and shear wavenumbers of the background medium, and
\begin{equation}%
\begin{array}
[c]{ccc}%
u_{p}:=-k_{p}^{-2}\mathrm{grad\,}\mathrm{div\,}u^{sc}\,, &  & u_{s}=k_{s}%
^{-2}\mathrm{curlcurl\,}u^{sc}\,,
\end{array}
\label{decomposition}%
\end{equation}
denote the longitudinal (compressional) and transversal (shear) parts of the
scattered field in $\Omega^{e}\subset{\mathbb{R}}^{3}$, respectively. Note
that in two dimensions the transversal (shear) part should be modified to be
\begin{equation}
{u}_{s}=k_{s}^{-2}\,\overrightarrow{\mbox{curl}}\,\mbox{curl}\;{u}^{sc},
\label{decomposition-2d}%
\end{equation}
where the two curl operators in ${\mathbb{R}}^{2}$ are defined by
\[
\overrightarrow{\mbox{curl}}\,v=\partial_{1}v_{2}-\partial_{2}v_{1},\quad
v=(v_{1},v_{2}),\qquad\mbox{curl}\;f:=(\partial_{2}f,-\partial_{1}f).
\]
It follows from the Navier equation (\ref{Navier}) and the decompositions
(\ref{decomposition})-(\ref{decomposition-2d}) that in $\Omega^{e}$,
\[
(\Delta+k_{\alpha}^{2})\,u_{\alpha}=0,\quad\alpha=p,s,\qquad\mathrm{div\,}%
u_{s}=0,\qquad\mbox{curl}\,u_{p}\;(\mbox{or}\;\overrightarrow{\mbox{curl}}%
\,u_{p})=0.
\]
The Kupradze radiation condition (\ref{Kupradze}) leads to the P-part
(longitudinal part) $u_{p}^{\infty}$ and the S-part (transversal part)
$u_{s}^{\infty}$ of the far-field pattern of $u^{sc}$, given by the asymptotic
behavior
\begin{equation}
u^{sc}(x)=\frac{\exp(ik_{p}|x|)}{|x|^{\frac{n-1}{2}}}\,u_{p}^{\infty}(\hat
{x})+\frac{\exp(ik_{s}|x|)}{|x|^{\frac{n-1}{2}}}\,u_{s}^{\infty}(\hat
{x})+\mathcal{O}(|x|^{-\frac{n+1}{2}}),\quad|x|\rightarrow+\infty, \label{far}%
\end{equation}
where, with some normalization, $u_{p}^{\infty}$ and $u_{s}^{\infty}$ are the
far-field patterns of $u_{p}$ and $u_{s}$, respectively. In this paper, we
define the far-field pattern $u^{\infty}$ of the scattered field $u^{sc}$ as
the sum of $u_{p}^{\infty}$ and $u_{s}^{\infty}$, that is, $u^{\infty}%
:=u_{p}^{\infty}+u_{s}^{\infty}$. It is well-known that $u_{p}^{\infty}$ is
normal to $\mathbb{S}^{n-1}$ and $u_{s}^{\infty}$ is tangential to
$\mathbb{S}^{n-1}$. Hence, we have the relations
\[
u_{p}^{\infty}(\hat{x})=(u^{\infty}(\hat{x})\cdot\hat{x})\,\hat{x},\quad
u_{s}^{\infty}(\hat{x})=\left\{
\begin{array}
[c]{lll}%
\hat{x}\times u^{\infty}(\hat{x})\times\hat{x},\quad &  & \mbox{if}\quad
n=3,\\
(u^{\infty}(\hat{x})\cdot\hat{x}^{\bot})\,\hat{x}^{\bot},\quad &  &
\mbox{if}\quad n=2,
\end{array}
\right.
\]
where $\hat{x}^{\bot}\in{\mathbb{S}}^{n-1}$ is perpendicular to $\hat{x}$.
Note that boundary and transmission conditions should be imposed on
$\partial\Omega$, relying on the physical property of $\Omega$. It is well
known that the forward scattering problem for both penetrable and impenetrable
bodies admits a unique solution $u\in H_{loc}^{1}(D^{c})$. To prove existence
of solutions we refer to \cite[Chapter 7.3]{Kupradze} for the standard
integral equation method applied to rigid scatterers with $C^{2}$-smooth
boundaries and to a recent paper \cite{BHSY2018} using the variational
approach for treating Lipschitz boundaries. Hence, the far-field pattern is
uniquely determined by the incident wave (for instance, exciting frequency and
direction) and the elastic body. Throughout our paper, an elastic body will be
referred as a point-like scatterer if its size is much smaller than the shear
wave length and the mass density has a high contrast of certain scale, as
compared to the background mass density, that will be described later. It is
called an extended scatterer if its size is comparable with the shear wave
length. We remark that the compressional wave length is greater than the shear
wave length in an isotropic and homogeneous medium. The aim of this paper is
to address the following direct and inverse problems:

\begin{itemize}
\item Describe the Foldy approach and the point-interaction model for elastic
scattering from finitely many point-like scatterers (see Section
\ref{sec:point} for the details).

\item Present a multi-scale model for elastic scattering by both point-like
scatterers and an extended rigid body (Section \ref{sec:multiscale}).

\item Recover multi-scale elastic scatterers from far-field patterns
corresponding to infinitely many plane waves with all directions excited at a
fixed frequency (Section \ref{inverse}).

\end{itemize}

In the presence of a finite number of point-like or small scatterers embedded
in a homogeneous medium, it is well known that the Born approximation models
the scattering effect by neglecting the wave interaction between these
scatterers. Consequently, the scattered field can be represented as a weighted
linear combination of point source waves emitting from each scatterer, where
the weight models the scattering strength (also called scattering
coefficients). Taking into account the multiple scattering, the Foldy formal
approach (see \cite{Foldy}) assumes that the scattering coefficient of each
scatterer is proportional to the external field acting on it (which is known
as the Foldy assumption) and suggests to present the scattered field as a
linear combination of the interactions between the point-like obstacles by
solving a linear algebraic system. From the mathematical point of view, the
solution to wave scattering from $M$ point-like obstacles can be rigorously
derived from the resolvent of a perturbed elliptic operator and the Krein's
inversion formula of the resolvents. In fact, point perturbation operator can
be regarded as the self-adjoint extension of some symmetric operator acting on
appropriate Sobolev spaces. For acoustic scattering from both point-like and
extended sound-soft obstacles, the point-interaction model was derived in
\cite{HMS} justifying the Foldy formel method and extending it to more general
models including the nonlocal interactions. A closed form of the
solution to such a multiscale scattering problem was obtained in \cite{HMS}.
Numerically, the authors of \cite{HL, HSZ} established an integral equation
representation based on the Foldy formal approach and proposed an iterative
approach for computing the unknown densities and coefficients.

The first aim of this paper is to justify the equivalence of the Foldy
approach and the point interaction model for the Lam\'e system. The extension
of our previous work \cite{HMS} to the linear elasticity turns out to be
non-trivial, mainly due the vectorial nature of the governing equation which
models a coupling of the propagation of compressional and shear waves.
Using the abstract construction of selfadjoint extensions by Posilicano \cite{Posi},
we model singular perturbations, of the Lam\'e operator, supported on a set of points, (see subsection \ref{BCM}).
This provides a generalized boundary conditions of impedance type on this set of points.
In the particular case of local and isotropic point perturbations,
we retrieve the closed form of the solution obtained in \cite{HM} by the renormalization techniques arising
from quantum mechanics \cite{AGHH}; see subsection \ref{subsec:point}.
 The multiscale point-interaction model for elastic scattering from a
combination of point-like and extended scatterers can be analogously
formulated. In Section \ref{sec:multiscale}, we present a straightforward
proof to the well-posedness of the resulting boundary value problem for
isotropic point interactions in linear elasticity. Related to our present work, let us mention the recent contribution \cite{BBPFN}, on point-like
perturbations for the two dimensional Lam\'{e} operator,
where the model is stated as a selfadjoint extension of a
symmetric restriction of $\Delta^{\ast}$ using boundary triplets. A
factorized representation of the fields is provided in the particular case of
local and isotropic perturbations.

The second aim of this paper is to investigate the inverse problem of imaging
an extended rigid elastic body and a finite number of point-like scatterers,.
We shall apply the factorization method \cite{Kirsch1998, K2008} by Kirsch to
such multi-scale inverse scattering problems by using different type of
elastic waves. In contrast to iterative schemes, the factorization method
requires neither direct solvers nor initial guesses, and it provides a
sufficient and necessary condition for characterizing the shape of the
extended obstacle and positions of the point-like scatterers. Note that there
is already a vast literature on inverse elastic scattering problems. The
linear-sampling and factorization methods were developed in \cite{AK, Arens}
and \cite{S2005, 2007} for imaging impenetrable and penetrable scatterers.
Using only one-type of elastic waves, uniqueness results for detecting extended scatterers
(penetrable or impenetrable) were proved in
\cite{GS, KS} and the MUSIC type algorithm \cite{GST} was applied to the
detection of point-like elastic scatterers. In \cite{HKS13}, the
factorization method was adapted to recover the shape of an extended rigid
body from the scattered S-waves (resp. P-waves) corresponding to all incident
plane shear (resp. pressure) waves. Within the framework of this paper, we
have unified the MUSIC algorithm for imaging point-like scatterers and the
classical factorization scheme for recovering extended obstacles.

The remaining part of this paper is organized as follows. In Section
\ref{pre}, we state properties of the resolvent and outgoing Green's tensor of
the Lam\'{e} operator in ${\mathbb{R}}^{n}$. Section \ref{sec:point} is
devoted to the Foldy approach and the point-interaction model for elastic
scattering by a collection of point-like scatterers. In Section
\ref{sec:multiscale}, we present mathematical formulations for the multi-scale
scattering problem and prove well-posedness of the boundary value problem.
Finally, the factorization method to inverse problems together with some
numerical tests are reported in Section \ref{inverse}.

\tcr{We end up this section by introducing some notation to be used later.
The spacial variables in $\R^n$ are denoted by $x=(x_1, x_2,\cdots, x_n)$ and $y=(y_1, y_2,\cdots, y_n)$, where $n=2,3$ is the spacial dimension.
 Denote by $\overline{(\cdot)}$ the closure of a set or the complex conjugate of a complex number. For $a\in \C$, let $|a|$ denote its modulus, and for
$\textbf{a}\in \C^2$, let $|\textbf{a}|$ denote its Euclidean norm. The symbol
$\textbf{a}\cdot\textbf{b}$
stands for the inner product
$a_1b_1+a_2b_2$ of $\textbf{a}=(a_1,a_2), \,\textbf{b}=(b_1,b_2)\in \C^2$.
Standard $L^2$-based scalar Sobolev spaces defined in a domain $D$ or on
a surface $M$ are denoted by $H^s(D)$ or $H^s(M)$ for $s\in\R$. By $\mathsf{B}(X,Y)$ we mean the space of bounded linear operators from the space $X$ to $Y$, and by $\mathbf{I}_{n}$ the identity matrix in $\R^n$. }

\section{Preliminaries}

\label{pre} \setcounter{equation}{0}


\subsection{Properties of the resolvent of the Lam\'e operator in
${\mathbb{R}}^{n}$.}

The quadratic form corresponding to the Lam\'e operator $-\Delta^{*}$ is given
by the closed form
\begin{equation}%
\begin{array}
[c]{ccc}%
Q_{0}\left(  u\right)  :=\lambda\left\Vert
{\textstyle\sum\nolimits_{i=1}^{n}}
\mathrm{div\,} u\right\Vert _{L^{2}\left(  \mathbb{R}^{n}\right)  }^{2}%
+\frac{\mu}{2}%
{\textstyle\sum\nolimits_{i,j=1}^{n}}
\left\Vert \partial_{i}u_{j}+\partial_{j}u_{i}\right\Vert _{L^{2}\left(
\mathbb{R}^{n}\right)  }^{2}\,, &  & u=(u_{1}, u_{2},\cdots, u_{n})
\end{array}
\label{Q_0}%
\end{equation}
with $\mathsf{dom}\left(  \bar{Q}_{0}\right)  =\left(  H^{1}\left(
\mathbb{R}^{n}\right)  \right)  ^{n}$. By (\ref{ConL}), it is positive defined
(see e.g. in \cite[Lemma 1.1]{ColdVe}). By \cite[Theorem VIII.15]{ReSi1},
there exists a unique selfadjoint operator $L_{0}$ on $\left(  L^{2}\left(
\mathbb{R}^{n}\right)  \right)  ^{n}$ fulfilling%
\[%
\begin{array}
[c]{ccc}%
Q_{0}\left(  u\right)  =\left\langle u,L_{0}u\right\rangle _{\left(
L^{2}\left(  \mathbb{R}^{n}\right)  \right)  ^{n}}\,, &  & u\in\mathsf{dom}%
\left(  L_{0}\right)  \,.
\end{array}
\]
This is the Friedrichs extension of $-\Delta^{*}$ and it is defined as%
\begin{equation}
\left\{
\begin{array}
[c]{l}%
\mathsf{dom}\left(  L_{0}\right)  :=\left(  H^{2}\left(  \mathbb{R}%
^{n}\right)  \right)  ^{n}\,,\\
\\
L_{0}u=-\mu\Delta u-\left(  \lambda+\mu\right)  \nabla\left(
\operatorname{div}u\right)  .
\end{array}
\right.  \label{Lame_op}%
\end{equation}
Since (\ref{Q_0}) is positive, $L_{0}$ is positive defined and we have the
resolvent set $\mathsf{res}\left(  L_{0}\right)  =\mathbb{C}\backslash\left[
0,+\infty\right)  $ and as $L_{0}:\left(  H^{2}\left(  \mathbb{R}^{n}\right)
\right)  ^{n}\rightarrow\left(  L^{2}\left(  \mathbb{R}^{n}\right)  \right)
^{n}$, it follows%
\begin{equation}%
\begin{array}
[c]{ccc}%
\mathcal{K}_{z}:=\left(  L_{0}-z\right)  ^{-1}\in\mathsf{B}\left(  \left(
L^{2}\left(  \mathbb{R}^{n}\right)  \right)  ^{n},\left(  H^{2}\left(
\mathbb{R}^{n}\right)  \right)  ^{n}\right)  \,, &  & z\in\mathsf{res}\left(
L_{0}\right)  \,.
\end{array}
\end{equation}
Denote the Laplacian resolvent by
\begin{equation}%
\begin{array}
[c]{ccc}%
R_{z}:=\left(  -\Delta-z\right)  ^{-1}\in\mathsf{B}\left(  H^{0}\left(
\mathbb{R}^{n}\right)  ,H^{2}\left(  \mathbb{R}^{n}\right)  \right)  \,, &  &
z\in\mathbb{C}\backslash\left[  0,+\infty\right)  \,.
\end{array}
\label{Res_laplacian}%
\end{equation}
The \emph{Kupradze matrix} defines the Lam\'e operator resolvent according to
(see \cite[Chp. 2]{Kup})%
\begin{equation}
\mathcal{K}_{z}=\frac{1}{\mu}R_{z/\mu}+\frac{1}{z}\nabla\operatorname{div}%
\left(  R_{z/\mu}-R_{z/\left(  \lambda+2\mu\right)  }\right)  \,,
\label{Kupradze_res}%
\end{equation}
where $\mathbf{I}_{n}$ is the identity matrix on $\mathbb{R}^{n}$, while
$z/\mu$ and $z/\left(  \lambda+2\mu\right)  $ are the rescaled energies
related to compressional and shared waves. The integral kernels of $R_{z}$ and
$\mathcal{K}_{z}$ are next denoted as%
\begin{equation}%
\begin{array}
[c]{ccc}%
\Phi_{z}\left(  x-y\right)  :=R_{z}\left(  x,y\right)  \,, &  & \Gamma
_{z}\left(  x-y\right)  :=\mathcal{K}_{z}\left(  x,y\right)  \,,
\end{array}
\label{Green_ker}%
\end{equation}
(where the identities $R_{z}^{\ast}=R_{\bar{z}}$ and $\mathcal{K}_{z}^{\ast
}=\mathcal{K}_{\bar{z}}$ are taken into account). From the identity
(\ref{Kupradze_res}), it follows that%
\begin{equation}
\Gamma_{z}\left(  x\right)  =\frac{1}{\mu}\Phi_{z/\mu}\left(  x\right)
+\frac{1}{z}\nabla\operatorname{div}\left(  \Phi_{z/\mu}\left(  x\right)
-\Phi_{z/\left(  \lambda+2\mu\right)  }\left(  x\right)  \right)  \,.
\label{Gamma_Res_ker}%
\end{equation}

We use the weighted spaces%
\begin{equation}
H_{\eta}^{s}\left(  \mathbb{R}^{n}\right)  :=\left\{  u\in\mathcal{D}^{\prime
}\left(  \mathbb{R}^{n}\right)  \,,\ \left\langle x\right\rangle ^{\eta}u\in
H^{s}\left(  \mathbb{R}^{n}\right)  \right\}  \,,\quad s\geq0\,,\ \eta
\in\mathbb{R}\,, \label{Sobolev_w}%
\end{equation}
where $\left\langle x\right\rangle ^{\eta}:=\left(  1+\left\vert x\right\vert
^{2}\right)  ^{\eta/2}$. The duals (w.r.t. the $L^{2}\left(  \mathbb{R}%
^{n}\right)  $ product) of (\ref{Sobolev_w}) are%
\begin{equation}
H_{-\eta}^{-s}\left(  \mathbb{R}^{n}\right)  :=\left\{  u\in\mathcal{D}%
^{\prime}\left(  \mathbb{R}^{n}\right)  \,,\ \left\langle x\right\rangle
^{-\eta}u\in H^{-s}\left(  \mathbb{R}^{n}\right)  \right\}  \,,\quad
s\geq0\,,\ \eta\in\mathbb{R}\,. \label{Sobolev_w_dual}%
\end{equation}
The Laplacian resolvent has well known mapping properties which are next
recalled. At first, we recall the resolvent identity%
\begin{equation}%
\begin{array}
[c]{ccc}%
R_{z}-R_{z_{0}}=\left(  z_{0}-z\right)  R_{z_{0}}R_{z}=\left(  z_{0}-z\right)
R_{z}R_{z_{0}}\,, &  & z,z_{0}\in\mathbb{C}\backslash\left[  0,+\infty\right)
\,.\,
\end{array}
\label{Res_id}%
\end{equation}
Using Fourier transform, duality and interpolation, from (\ref{Res_laplacian})
follows $R_{z}\in\mathsf{B}\left(  H^{s}\left(  \mathbb{R}^{n}\right)
,H^{2+s}\left(  \mathbb{R}^{n}\right)  \right)  $, for any $s\in\mathbb{R}$,
and%
\begin{equation}%
\begin{array}
[c]{ccc}%
\left\Vert R_{z}\right\Vert _{H^{s}\left(  \mathbb{R}^{n}\right)
,H^{s+t}\left(  \mathbb{R}^{n}\right)  }\leq\frac{1}{d^{1-t/2}\left(
z,\left[  0,+\infty\right)  \right)  }\,, &  & t\in\left[  0,2\right]  \,,
\end{array}
\label{R_est}%
\end{equation}
where $d\left(  \cdot,\left[  0,+\infty\right)  \right)  $ is the distance
from the set $\left[  0,+\infty\right)  $. According to \cite[Lemma 1,
p.170]{ReSi IV}, one has: $R_{z}\in\mathsf{B}\left(  L_{\eta}^{2}\left(
\mathbb{R}^{n}\right)  \right)  $, for any $\eta\in\mathbb{R}$; this entails
(see \cite[relation (4.8)]{MaPoSI LAP}) that $R_{z}\in\mathsf{B}\left(
L_{\eta}^{2}\left(  \mathbb{R}^{n}\right)  ,H_{\eta}^{2}\left(  \mathbb{R}%
^{n}\right)  \right)  $ and, by duality and interpolation, we get%
\begin{equation}%
\begin{array}
[c]{ccccc}%
R_{z}\in\mathsf{B}\left(  H_{\eta}^{-s}\left(  \mathbb{R}^{n}\right)
,H_{\eta}^{2-s}\left(  \mathbb{R}^{n}\right)  \right)  \,, &  & \eta
\in\mathbb{R}\,, &  & s\in\left[  -2,0\right]  \,.
\end{array}
\end{equation}
Since $H_{\eta}^{s}\left(  \mathbb{R}^{n}\right)  \hookrightarrow H^{s}\left(
\mathbb{R}^{n}\right)  $ for $\eta>0$, the previous mapping properties also
yield%
\begin{equation}%
\begin{array}
[c]{ccccc}%
R_{z}\in\mathsf{B}\left(  H_{\eta}^{s}\left(  \mathbb{R}^{n}\right)
,H_{-\eta}^{2+s}\left(  \mathbb{R}^{n}\right)  \right)  \,, &  & \eta>0\,, &
& s\in\mathbb{R}\,.
\end{array}
\end{equation}
Moreover, it is well-known that $z\rightarrow R_{z}$ is holomorphic in
$z\in\mathbb{C}\backslash\left[  0,+\infty\right)  $ and that a limiting
absorption principle holds (see e.g. \cite[Theorem 4.1]{Agm}, \cite[Theorem
18.3]{KoKo}), i.e. the limits%
\begin{equation}%
\begin{array}
[c]{ccc}%
R_{\omega^{2}}^{\pm}:=\lim_{\varepsilon\rightarrow0^{+}}R_{\omega^{2}\pm
i\varepsilon}\,, &  & \omega>0\,,
\end{array}
\end{equation}
exist in $\mathsf{B}\left(  L_{\eta}^{2}\left(  \mathbb{R}^{n}\right)
,H_{-\eta}^{2}\left(  \mathbb{R}^{n}\right)  \right)  $ with $\eta>1/2$ and
they satisfy%
\begin{equation}
\left(  -\Delta-\omega^{2}\right)  R_{\omega^{2}}^{\pm}=1\,.
\end{equation}
This limit allows us to define the extended map%
\begin{equation}
z\rightarrow R_{z}^{\pm}:=\left\{
\begin{array}
[c]{lll}%
R_{z}\,, &  & z\in\mathbb{C}\backslash\left[  0,+\infty\right)  \,,\\
&  & \\
R_{\omega^{2}}^{\pm}\,, &  & z=\omega^{2}\pm i0\,.
\end{array}
\right.  \label{Res_ext}%
\end{equation}
Using this definition, the resolvent identity extends according to%
\begin{equation}%
\begin{array}
[c]{ccccc}%
R_{z}^{\pm}-R_{z_{0}}=\left(  z_{0}-z\right)  R_{z_{0}}R_{z}^{\pm}\,, &  &
z\in\mathbb{C}\,, &  & z_{0}\in\mathbb{C}\backslash\left[  0,+\infty\right)
\,.
\end{array}
\end{equation}
Using this relation, it is easy to verify by iteration, duality and
interpolation, that the limit mapping properties improve as%
\begin{equation}%
\begin{array}
[c]{ccccc}%
R_{z}^{\pm}\in\mathsf{B}\left(  H_{\eta}^{s}\left(  \mathbb{R}^{n}\right)
,H_{-\eta}^{2+s}\left(  \mathbb{R}^{n}\right)  \right)  \,, &  & \eta>0\,, &
& s\in\mathbb{R}\,.
\end{array}
\label{Res_ext_map}%
\end{equation}

The above properties of the Laplace operator extend to the Lam\'e operator as
well. We state them in the following theorem.

\begin{theorem}
\label{Thorem_res}$i)$ Let $s\in\mathbb{R}$. For $z\in\mathbb{C}%
\backslash\left[  0,+\infty\right)  $ the map $z\rightarrow\mathcal{K}_{z}$ is
holomorphic with values in\newline$\mathsf{B}\left(  \left(  H^{s}\left(
\mathbb{R}^{n}\right)  \right)  ^{n},\left(  H^{s+2}\left(  \mathbb{R}%
^{n}\right)  \right)  ^{n}\right)  $ and fulfills the estimates%
\begin{equation}%
\begin{array}
[c]{ccc}%
\left\Vert \mathcal{K}_{z}\right\Vert _{H^{s}\left(  \mathbb{R}^{n}\right)
,H^{s+t}\left(  \mathbb{R}^{n}\right)  }\leq\frac{1}{d^{1-t/2}\left(
z,\left[  0,+\infty\right)  \right)  }\,, &  & t\in\left[  0,2\right]  \,.
\end{array}
\label{Kupradze_res_est}%
\end{equation}
$ii)$ If $\eta>1/2$, $z\rightarrow R_{z}$ continuously extends to the limits
$z\rightarrow\omega^{2}\pm i0$, $\omega>0$ in the weaker topology of
$\mathsf{B}\left(  \left(  H_{\eta}^{s}\left(  \mathbb{R}^{n}\right)  \right)
^{n},\left(  H_{-\eta}^{s+2}\left(  \mathbb{R}^{n}\right)  \right)
^{n}\right)  $, i.e. the limits%
\begin{equation}
\mathcal{K}_{\omega^{2}}^{\pm}:=\frac{1}{\mu}R_{k_{s}^{2}}^{\pm}+\frac
{1}{\omega^{2}}\nabla\operatorname{div}\left(  R_{k_{s}^{2}}^{\pm}%
-R_{k_{p}^{2}}^{\pm}\right)  =\lim_{\varepsilon\rightarrow0^{+}}%
\mathcal{K}_{\omega^{2}\pm i\varepsilon}\,, \label{Kupradze_res_lim}%
\end{equation}
exist in $\mathsf{B}\left(  \left(  H_{\eta}^{s}\left(  \mathbb{R}^{n}\right)
\right)  ^{n},\left(  H_{-\eta}^{s+2}\left(  \mathbb{R}^{n}\right)  \right)
^{n}\right)  $ with $\eta>1/2$, $s\in\mathbb{R}$, and they satisfy%
\begin{equation}
\left(  -L_{0}-\omega^{2}\right)  \mathcal{K}_{\omega^{2}}^{\pm}%
=\mathbf{I}_{n}\,. \label{Kupradze_res_lim_id}%
\end{equation}

\end{theorem}

\begin{proof}
$i)$ The mapping properties of $z\rightarrow R_{z}$ imply that $z\rightarrow
R_{z/\mu}R_{z/\left(  \lambda+2\mu\right)  }$ is a $\mathsf{B}\left(  \left(
H^{s}\left(  \mathbb{R}^{n}\right)  \right)  ^{n},\left(  H^{s+4}\left(
\mathbb{R}^{n}\right)  \right)  ^{n}\right)  $-valued holomorphic function in
$z\in\mathbb{C}\backslash\left[  0,+\infty\right)  $; by (\ref{Res_id}), we
have%
\begin{equation}
\frac{1}{z}\nabla\operatorname{div}\left(  R_{z/\mu}-R_{z/\left(  \lambda
+2\mu\right)  }\right)  =-\frac{\lambda+\mu}{\mu\left(  \lambda+2\mu\right)
}\nabla\operatorname{div}\left(  R_{z/\mu}R_{z/\left(  \lambda+2\mu\right)
}\right)  \,,
\end{equation}
and (\ref{Kupradze_res}) rephrases as%
\begin{equation}
\mathcal{K}_{z}=\frac{1}{\mu}R_{z/\mu}-\frac{\lambda+\mu}{\mu\left(
\lambda+2\mu\right)  }\nabla\operatorname{div}\left(  R_{z/\mu}R_{z/\left(
\lambda+2\mu\right)  }\right)  \,.
\end{equation}
Hence the first statement follows from the mapping properties of $R_{z}$
recalled above. In particular the estimates (\ref{Kupradze_res_est}) follow
from (\ref{Kupradze_res}) and (\ref{R_est}). $ii)$ Let $z_{0},z_{1},z_{2}%
\in\mathbb{C}\backslash\left[  0,+\infty\right)  $; using twice (\ref{Res_id})
we have%
\begin{equation}
R_{z_{1}}-R_{z_{2}}=R_{z_{0}}\left(  \left(  z_{0}-z_{1}\right)  R_{z_{1}%
}-\left(  z_{0}-z_{2}\right)  R_{z_{2}}\right)  \,.
\end{equation}
By (\ref{Res_ext_map}), the r.h.s. continuously extends both to the limits
$z_{j=1,2}\rightarrow\omega^{2}\pm i0$ as a $\mathsf{B}\left(  \left(
H_{\eta}^{s}\left(  \mathbb{R}^{n}\right)  \right)  ^{n},\left(  H_{-\eta
}^{s+4}\left(  \mathbb{R}^{n}\right)  \right)  ^{n}\right)  $-valued map, for
$\eta>1/2$. Hence, $z\rightarrow\nabla\operatorname{div}\left(  R_{z_{s}%
}-R_{z_{p}}\right)  $ is continuous in $\mathbb{C}\backslash\left[
0,+\infty\right)  $ up to $z\rightarrow\omega^{2}\pm i0$ with values in
$\mathsf{B}\left(  \left(  H_{\eta}^{s}\left(  \mathbb{R}^{n}\right)  \right)
^{n},\left(  H_{-\eta}^{s+2}\left(  \mathbb{R}^{n}\right)  \right)
^{n}\right)  $ and the limits are defined by%
\begin{equation}
R_{k_{s}^{2}}^{\pm}-R_{k_{p}^{2}}^{\pm}=R_{z_{0}}\left(  \left(  z_{0}%
-k_{s}^{2}\right)  R_{k_{s}^{2}}^{\pm}-\left(  z_{0}-k_{p}^{2}\right)
R_{k_{p}^{2}}^{\pm}\right)  \,.
\end{equation}
From (\ref{Kupradze_res}), the limits (\ref{Kupradze_res_lim}) hold. Finally,
the limits (\ref{Kupradze_res_lim_id}) follows from%
\begin{equation}%
\begin{array}
[c]{ccc}%
\left(  -L_{0}-z\right)  \mathcal{K}_{z}u=u\,, &  & u\in\left(  H_{\eta}%
^{s}\left(  \mathbb{R}^{n}\right)  \right)  ^{n}\,.
\end{array}
\end{equation}

\end{proof}

The Green kernels of the limiting operators $R_{k^{2}}^{\pm}$, expressed by
the limits%
\begin{equation}
\Phi_{k^{2}}^{\pm}:=\lim_{\varepsilon\rightarrow0^{+}}\Phi_{k^{2}\pm
i\varepsilon}\,,
\end{equation}
are radiating solutions of the Helmoholtz equation in $\mathbb{R}%
\backslash\left\{  0\right\}  $, i.e. these satisfy the radiation conditions%
\begin{equation}
\lim_{r \rightarrow0}r^{\left(  n-1\right)  /2}\left(  \partial_{r}\mp i
k\right)  \Phi_{k^{2}}^{\pm}\left(  x\right)  =0\,,\quad r=|x|.
\end{equation}
From (\ref{Kupradze_res_lim}), the corresponding limit kernels of
$\mathcal{K}_{\omega^{2}}^{\pm}$ are given by%
\begin{equation}
\Gamma_{\omega^{2}}^{\pm}:=\frac{1}{\mu}\Phi_{k_{s}}^{\pm}+\frac{1}{\omega
^{2}}\nabla\operatorname{div}\left(  \Phi_{k_{s}}^{\pm}-\Phi_{k_{p}}^{\pm
}\right)  \,,
\end{equation}
where we recall that $k_{p}$ and $k_{s}$ are the compressional and shear
wavenumbers, respectively. Following \cite{Kup}, these are solutions of the
Lam\'{e} stationary equation $\left(  L_{0}-\omega^{2}\right)  \Gamma
_{\omega^{2}}^{\pm}=0$ in $\mathbb{R}\backslash\left\{  0\right\}  $ and, for
$\omega\neq0$, fulfill the Kupradze radiation conditions%
\begin{align}
&  \left.  \lim_{r \rightarrow0}r^{\left(  n-1\right)  /2}\left(  \partial
_{r}\mp i\,k_{p}\right)  \nabla\operatorname{div}\Gamma_{\omega^{2}}^{\pm
}=0\,,\right. \label{rad_cond_L}\\
&  \left.  \lim_{r\rightarrow0}r^{\left(  n-1\right)  /2}\left(  \partial
_{r}\mp i\,k_{s}\right)  \nabla\times\nabla\times\Gamma_{\omega^{2}}^{\pm
}=0\,.\right.  \label{rad_cond_T}%
\end{align}

\subsection{\label{Sec_Green_2D}Outgoing Green's tensor in 2D}

We recall some known properties of the integral kernels $\Phi_{z}$ and
$\Gamma_{z}$ for $n=2$; more details can be found in \cite{Kup}, \cite{HuSi}
and references therein. If $n=2$, the integral kernel of $R_{z}$ is given by%
\begin{equation}
\Phi_{z}\left(  x-y\right)  :=\frac{i}{4}H_{0}^{\left(  1\right)  }\left(
\zeta\left\vert x-y\right\vert \right)  \,,\qquad\zeta\in\mathbb{C}^{+}%
:\zeta^{2}=z\in\mathbb{C}\backslash\left[  0,+\infty\right)  \,,
\label{Green_ker_2D}%
\end{equation}
where $H_{0}^{\left(  1\right)  }$ is the Henkel function of the first kind
and of order zero. By (\ref{Gamma_Res_ker}) and (\ref{Green_ker_2D}), the
$\mathcal{C}^{\infty}\left(  \mathbb{R}^{2,2}\right)  $-valued map
$z\rightarrow\Gamma_{z}\left(  x\right)  $, introduced in (\ref{Gamma_Res_ker}%
), is holomorphic in $z\in\mathbb{C}\backslash\left[  0,+\infty\right)  $ and
continuously extends to the limits $z\rightarrow\omega^{2}\pm i0$,
$k\in\mathbb{R}$.; in particular, by \cite[eq. (11)]{HuSi}, the limit%
\begin{equation}
\Gamma_{0}\left(  x\right)  :=\lim_{z\rightarrow0}\Gamma_{z}\left(  x\right)
=\frac{1}{4\pi}\left[  -\frac{\lambda+3\mu}{\mu\left(  \lambda+2\mu\right)
}\ln\left\vert x\right\vert
\begin{pmatrix}
1 & 0\\
0 & 1
\end{pmatrix}
+\frac{\lambda+\mu}{\mu\left(  \lambda+2\mu\right)  }\frac{1}{\left\vert
x\right\vert ^{2}}%
\begin{pmatrix}
x_{1}^{2} & x_{1}x_{2}\\
x_{1}x_{2} & x_{2}^{2}%
\end{pmatrix}
\right]  \,, \label{Gamma_Res_ker_2D_lim}%
\end{equation}
pointwise holds in $\mathbb{R}^{2}\backslash\left\{  0\right\}  $. From
(\ref{Gamma_Res_ker_2D_lim}) we get%
\begin{equation}
\det\left(  \Gamma_{0}\left(  x\right)  \right)  =\frac{\ln\left\vert
x\right\vert \left(  \lambda+3\mu\right)  }{\left(  4\pi\mu\left(
\lambda+2\mu\right)  \right)  ^{2}}\left(  \left(  \lambda+\mu\right)
-\left(  \lambda+3\mu\right)  \ln\left\vert x\right\vert \right)  \,.
\end{equation}
Hence the inverse matrix $\Gamma_{0}^{-1}\left(  x\right)  $ exists in
$\mathbb{R}^{2}\backslash\left\{  0,\left\vert x\right\vert =e^{\frac
{\lambda+\mu}{\lambda+3\mu}}\right\}  $ where it writes as%
\be
\Gamma_{0}^{-1}\left(  x\right)  &=& \frac{\left(  4\pi\mu\left(  \lambda+2\mu\right)  \right)  ^{2}%
}{\ln\left\vert x\right\vert \left(  \lambda+3\mu\right)  \left(  \left(
\lambda+\mu\right)  -\left(  \lambda+3\mu\right)  \ln\left\vert x\right\vert
\right)  }
\nonumber\\
&\quad& \times
\begin{pmatrix}
\left(  \lambda+3\mu\right)  \ln\left\vert x\right\vert -\left(  \lambda
+\mu\right)  \frac{x_{2}^{2}}{\left\vert x\right\vert ^{2}} & -\left(
\lambda+\mu\right)  \frac{x_{1}x_{2}}{\left\vert x\right\vert ^{2}}\\
-\left(  \lambda+\mu\right)  \frac{x_{1}x_{2}}{\left\vert x\right\vert ^{2}} &
\left(  \lambda+3\mu\right)  \ln\left\vert x\right\vert -\left(  \lambda
+\mu\right)  \frac{x_{1}^{2}}{\left\vert x\right\vert ^{2}}%
\end{pmatrix}
\,. \label{Gamma_lim_inv_2D}%
\en
It follows%
\begin{equation}%
\begin{array}
[c]{ccc}%
\left\vert \Gamma_{0}^{-1}\left(  x\right)  \right\vert _{\mathbb{R}^{2,2}%
}=\mathcal{O}\left(  1/\ln\left\vert x\right\vert \right)  \,, &  & \text{as
}x\rightarrow0\,.
\end{array}
\label{Gamma_lim_inv_2D_est}%
\end{equation}
Let%
\begin{equation}
\tilde{\Gamma}_{z}\left(  x\right)  :=\Gamma_{z}\left(  x\right)  -\Gamma
_{0}\left(  x\right)  \,; \label{Gamma_tilda}%
\end{equation}
by \cite[Lemma 2.1]{HuSi} it results%
\begin{equation}
\left\{
\begin{array}
[c]{l}%
\chi_{z}\,\mathbf{I}_{2}:=\lim_{\left\vert x\right\vert \rightarrow0}%
\tilde{\Gamma}_{z}\left(  x\right)  \,,\\
\\
\chi_{z}=\chi_{z}(\lambda,\mu)=-\frac{1}{4\pi}\left[  \frac{\lambda+3\mu}%
{\mu\left(  \lambda+2\mu\right)  }\left(  \ln\frac{\sqrt{z}}{2}+C-i\frac{\pi
}{2}\right)  +\frac{\lambda+\mu}{\mu\left(  \lambda+2\mu\right)  }-\frac{1}%
{2}\left(  \frac{\ln\mu}{\mu}+\frac{\ln\left(  \lambda+2\mu\right)  }%
{\lambda+2\mu}\right)  \right]  \,,
\end{array}
\right.  \label{Gamma_tilda_lim_2D}%
\end{equation}
where $C$ is the Euler constant and $\sqrt{z}$ is defined with
$\operatorname{Im}\sqrt{z}>0$. From (\ref{Gamma_tilda_lim_2D}) we get%
\begin{equation}%
\begin{array}
[c]{ccc}%
\left\vert \Gamma_{0}^{-1}\left(  x\right)  \tilde{\Gamma}_{z}\left(
x\right)  \right\vert _{\mathbb{R}^{2,2}}=\mathcal{O}\left(  1/\ln\left\vert
x\right\vert \right)  \,, &  & \text{as }x\rightarrow0\,.
\end{array}
\label{Gamma_lim_inv_2D_est_1}%
\end{equation}
By (\ref{Gamma_lim_inv_2D_est})-(\ref{Gamma_lim_inv_2D_est_1}) we get%
\begin{equation}
\Gamma_{0}^{-1}\left(  x\right)  \Gamma_{z}\left(  y\right)  =\Gamma_{0}%
^{-1}\left(  x\right)  \left(  \Gamma_{0}\left(  y\right)  +\tilde{\Gamma}%
_{z}\left(  y\right)  \right)  =\left\{
\begin{array}
[c]{lll}%
\mathbf{I}_{2}+\mathcal{O}\left(  1/\ln\left\vert x\right\vert \right)  \,, &
& \text{as }x\rightarrow0\text{ if }x=y\,,\\
&  & \\
o\left(  1/\ln\left\vert x\right\vert \right)  \,, &  & \text{as }%
x\rightarrow0\text{ if }x\neq y\,,
\end{array}
\right.  \label{Gamma_lim_inv_2D_est_2}%
\end{equation}
holding in the $\mathbb{R}^{2,2}$-norm sense.

\subsection{\label{Sec_Green_3D}Outgoing Green's tensor in 3D}

If $n=3$, the integral kernel of $R_{z}$ is explicitly given by%
\begin{equation}
\Phi_{z}\left(  x-y\right)  =\frac{e^{i\zeta\left\vert x-y\right\vert }}%
{4\pi\left\vert x-y\right\vert }\,,\qquad\zeta\in\mathbb{C}^{+}:\zeta^{2}%
=z\in\mathbb{C}\backslash\left[  0,+\infty\right)  \,. \label{Green_ker_3D}%
\end{equation}
By (\ref{Gamma_Res_ker}) and (\ref{Green_ker_3D}), the $\mathcal{C}^{\infty
}\left(  \mathbb{R}^{3,3}\right)  $-valued map $z\rightarrow\Gamma_{z}\left(
x\right)  $ is holomorphic in $z\in\mathbb{C}\backslash\left[  0,+\infty
\right)  $ and continuously extends to the limits $z\rightarrow k\pm i0$,
$k\in\mathbb{R}$.; by (\cite[eq. 48]{HuSi}), the limit%
\begin{equation}
\left(  \Gamma_{0}\left(  x\right)  \right)  _{j,\ell}:=\lim_{z\rightarrow
0}\left(  \Gamma_{z}\left(  x\right)  \right)  _{j,\ell}=\frac{\lambda+3\mu
}{8\pi\mu\left(  \lambda+2\mu\right)  }\frac{\delta_{j,\ell}}{\left\vert
x\right\vert }+\frac{\lambda+\mu}{8\pi\mu\left(  \lambda+2\mu\right)  }%
\frac{x_{j}x_{\ell}}{\left\vert x\right\vert ^{3}}\,,
\label{Gamma_Res_ker_3D_lim}%
\end{equation}
pointwise holds in $\mathbb{R}^{3}\backslash\left\{  0\right\}  $. From%
\begin{equation}
\det\left(  \frac{8\pi\mu\left(  \lambda+2\mu\right)  \left\vert x\right\vert
^{3}}{\lambda+\mu}\Gamma_{0}\left(  x\right)  \right)  =\left(  \left(
\frac{2\mu}{\lambda+\mu}\right)  ^{2}+3\frac{2\mu}{\lambda+\mu}+2\right)
\left\vert x\right\vert ^{2}>0\quad\text{in }\mathbb{R}^{3}\backslash\left\{
0\right\}  \,.
\end{equation}
This allows to define the inverse tensor $\Gamma_{0}^{-1}\left(  x\right)  $
which writes as%
\begin{equation}%
\begin{array}
[c]{ccc}%
\left(  \Gamma_{0}^{-1}\left(  x\right)  \right)  _{j,\ell}=q\left(
\lambda,\mu\right)  \left\vert x\right\vert \left(  2\left(  1+\frac{\mu
}{\lambda+\mu}\right)  \left\vert x\right\vert ^{2}\delta_{j,\ell}%
-x_{j}x_{\ell}\right)  \,, &  & q\left(  \lambda,\mu\right)  :=4\pi\mu
\frac{\lambda+\mu}{\lambda+3\mu}\,.
\end{array}
\label{Gamma_lim_inv_3D}%
\end{equation}
It follows%
\begin{equation}%
\begin{array}
[c]{ccc}%
\left\vert \Gamma_{0}^{-1}\left(  x\right)  \right\vert _{\mathbb{R}^{3,3}%
}=o\left(  \left\vert x\right\vert \right)  \,, &  & \text{as }x\rightarrow
0\,.
\end{array}
\label{Gamma_lim_inv_3D_est}%
\end{equation}
As before we denote%
\begin{equation}
\tilde{\Gamma}_{z}\left(  x\right)  :=\Gamma_{z}\left(  x\right)  -\Gamma
_{0}\left(  x\right)  \,. \label{Gamma_tilda_3D}%
\end{equation}
By (\cite[eq. 49]{HuSi}) results%
\begin{equation}
\chi_{z}\,\mathbf{I}_{3}:=\lim_{x\rightarrow0}\tilde{\Gamma}_{z}\left(
x\right)  =i\sqrt{z}\frac{2\lambda+5\mu}{12\pi\mu\left(  \lambda+2\mu\right)
}\,\mathbf{I}_{3}\,, \label{Gamma_tilda_lim_3D}%
\end{equation}
and from (\ref{Gamma_lim_inv_3D_est}) we get%
\begin{equation}%
\begin{array}
[c]{ccc}%
\left\vert \Gamma_{0}^{-1}\left(  x\right)  \tilde{\Gamma}_{\zeta}\left(
x\right)  \right\vert _{\mathbb{R}^{3,3}}=o\left(  \left\vert x\right\vert
\right)  \,, &  & \text{as }x\rightarrow0\,.
\end{array}
\label{Gamma_lim_inv_3D_est_1}%
\end{equation}
It follows%
\begin{equation}
\Gamma_{0}^{-1}\left(  x\right)  \Gamma_{\zeta}\left(  y\right)  =\Gamma
_{0}^{-1}\left(  x\right)  \left(  \Gamma_{0}\left(  y\right)  +\tilde{\Gamma
}_{\zeta}\left(  y\right)  \right)  =\left\{
\begin{array}
[c]{lll}%
\mathbf{I}_{3}+o\left(  \left\vert x\right\vert \right)  \,, &  & \text{as
}x\rightarrow0\text{ if }x=y\,,\\
&  & \\
o\left(  \left\vert x\right\vert \right)  \,, &  & \text{as }x\rightarrow
0\text{ if }x\neq y\,,
\end{array}
\right.  \label{Gamma_lim_inv_3D_est_2}%
\end{equation}
holding in the $\mathbb{R}^{3,3}$-norm sense.

\section{\label{sec:point}Elastic scattering by a collection of point-like
obstacles}

\setcounter{equation}{0} In this section we consider the time-harmonic elastic
scattering by $N$ point-like scatterers located at $y^{(k)},k=1,\cdots,N$ in
${\mathbb{R}}^{n}$ $(n=2,3)$. The set of these point-like scatterers will be
denoted by $Y:=\{y^{(k)}:k=1,\cdots,N\}$. Physically, such small obstacles are
related to highly concentrated inhomogeneous elastic medium (i.e. the mass
density in our case) with sufficiently small diameters compared to the
wave-length of incidence. In other words we shall suppose that
\[%
\begin{array}
[c]{cccc}%
\Omega=\bigcup_{k=1}^{N}D_{k}\,,~~\rho\lvert_{D_{k}}\approx(\mbox{diam}(D_{k}%
)^{-n}\;\mbox{ and } & \frac{\omega}{2\pi}\mbox{diam}(D_{k})\ll1\,, &  &
k=1,\cdots,N\,.
\end{array}
\]

Note that the wave-length for compressional and shear waves are defined via
\begin{equation}\label{wavelength}
\lambda_{p}:=\frac{2\pi}{k_{p}}=\frac{2\pi}{\omega}\sqrt{\lambda+2\mu}%
,\quad\lambda_{s}:=\frac{2\pi}{k_{s}}=\frac{2\pi}{\omega}\sqrt{\mu}%
\end{equation}
respectively. From mathematical point of view, the presence of these
point-like obstacles corresponds to a formal Delta-like perturbation of the
density function
\begin{equation}
\rho(x)-1=\sum_{k=1}^{N}\,a_{k}\,\delta(x-y^{(k)})\,, \label{delta}%
\end{equation}
where $a_{k}\in{\mathbb{C}}$ is the scattering strength (coupling constant)
attached to the scatter located at $y^{(k)}$. We write $a=(a_{1},a_{2}%
,\cdots,a_{N})\in{\mathbb{C}}^{N}$ and denote by $\mathbf{I}_{N}$ the identity
matrix in ${\mathbb{C}}^{N\times N}$.

\subsection{Foldy approach}

A formal solution to the scattering problem (\ref{Navier}),(\ref{Kupradze})
and (\ref{delta}) is given by
\begin{equation}
u(x)=u^{in}(x)+\sum_{m=1}^{N}a_{m}\,\Gamma_{\omega^{2}}(x,y^{(m)}%
)\,u(y^{(m)}),\quad x\neq y^{(m)},\qquad m=1,2,\cdots,N, \label{solution}%
\end{equation}
where $\Gamma_{\omega^{2}}(x,y)$ is the fundamental tensor of the Lam\'{e}
operator. To determine the value of $u$ at $y^{(k)}$ on the right hand side of
(\ref{solution}), the Foldy approach, originated in acoustic scattering from
many particles \cite{Foldy,VCL1998, Martin2006}, suggests solving the linear
algebraic system
\begin{equation}
u(y^{(k)})=u^{in}(y^{(k)})+\sum_{m=1,m\neq j}^{N}a_{m}\,\Gamma_{\omega^{2}%
}(y^{(k)},y^{(m)})\,u(y^{(m)}),\quad k=1,2,\cdots,N. \label{Foldy}%
\end{equation}
In fact, the Foldy system (\ref{Foldy}) follows from taking the limits
$x\rightarrow y^{(k)}$ in (\ref{solution}) and removing the singular term
(i.e., when $m=k$) in the sum. It has been shown in (\cite{C-S}) that the
algebraic system is uniquely solvable except for a discrete set of frequencies
and some particular distribution of the points at $y^{(k)}$. Inserting the
solution of (\ref{Foldy}) into (\ref{solution}) we obtain an explicit
representation of the solution of our scattering problem in the form of
\begin{equation}
u(x)=u^{in}+\sum_{m,k=1}^{N}\Gamma_{\omega^{2}}(x,y^{(k)})\;[\Pi_{\omega^{2}%
}^{-1}]_{m,k}\;u^{in}(y^{(m)}), \label{solution1}%
\end{equation}
where $[\Pi_{\omega^{2}}^{-1}]_{m,k}$ denotes the $(m,k)$-th entry of the
inverse of the block-matrix $\Pi_{\omega^{2}}$ given by
\[
\Pi_{\omega^{2}}:=%
\begin{pmatrix}
\mathbf{I}_{n} & -a_{2}\,\Gamma_{\omega^{2}}(y^{(1)},y^{(2)}) & \cdots &
-a_{N}\,\Gamma_{\omega^{2}}(y^{(1)},y^{(N)})\\
-a_{1}\,\Gamma_{\omega^{2}}(y^{(2)},y^{(1)}) & \mathbf{I}_{n} & \cdots &
-a_{N}\,\Gamma_{\omega^{2}}(y^{(2)},y^{(N)})\\
\vdots & \vdots & \ddots & \vdots\\
-a_{N}\,\Gamma_{\omega^{2}}(y^{(N)},y^{(1)}) & -a_{2}\,\Gamma_{\omega^{2}%
}(y^{(2)},y^{(N)}) & \cdots & \mathbf{I}_{n}%
\end{pmatrix}
.
\]
A rigorous justification of the Foldy system (\ref{solution1}) was carried out
in \cite{HM} by applying the \emph{renormalization techniques} arising from
quantum mechanics for describing the point interaction of finitely many
particles. Replacing the scattering coefficients $a_{k}$ by parameter
dependent functions $ak(\epsilon)$ that decay in a suitable way as
$\epsilon\rightarrow0$, one can show via Weinstein-Aronszajn determinant
formula that the resolvent of the a family of $\epsilon$-dependent delta
perturbations of the Lam\'{e} operator converges in Agmon's weighted spaces.
The resolvent of the limiting operator leads to the same expression of $u$ as
in (\ref{solution1}) with
\[
a_{k}=b_{k}-\chi_{\omega^{2}},\quad b_{k}\in{\mathbb{C}},\qquad k=1,\cdots,N,
\]
where $\chi_{\omega^{2}}\in{\mathbb{C}}$ are the dimension-dependent constants
given by (\ref{Gamma_tilda_lim_2D}) and (\ref{Gamma_tilda_lim_3D}) with
$z=\omega^{2}$. These quantities corresponds to the the normalizing constants
introduced in \cite[relations (12) and (49)]{HuSi}.

\subsection{Point-interaction of the Lam\'e operator}

\label{sec:point-interaction} In this study, the scattering effect due to the
presence of point-like obstacles is modeled as elastic point interactions.
Within this model, the scattering problem will be regarded as singular
perturbations supported on points.


We next consider a collection of $n$ points $Y=\left\{  y^{(k)}\right\}
_{k=1}^{N}\subset\mathbb{R}^{n}$; by the Sobolev imbedding $\mathcal{C}%
^{1}\hookrightarrow H_{\eta}^{2}\left(  \mathbb{R}^{n}\right)  $ holding for
$n\leq3$, $\eta\in\mathbb{R}$, the auxiliary map%
\begin{equation}
\gamma:\left(  H_{\eta}^{2}\left(  \mathbb{R}^{n}\right)  \right)
^{n}\rightarrow\mathbb{C}^{n,N}\,,\quad\left(  \gamma u\right)  _{j,k}%
:=u_{j}\left(  y^{(k)}\right)  \,,\ j=1,..n\,,\ k=1,..N\,, \label{trace}%
\end{equation}
is continuous and surjective. From the identity
\begin{equation}
\left\langle \gamma\varphi,m\right\rangle _{\mathbb{R}^{n}}=\left\langle
\varphi,\gamma^{\ast}m\right\rangle _{\left(  H^{2}\left(  \mathbb{R}%
^{n}\right)  \right)  ^{n},\left(  H^{-2}\left(  \mathbb{R}^{n}\right)
\right)  ^{n}}\,,\quad\varphi\in\left(  H_{\eta}^{2}\left(  \mathbb{R}%
^{n}\right)  \right)  ^{n}\,,
\end{equation}
follows that%
\begin{equation}%
\begin{array}
[c]{ccc}%
\gamma^{\ast}\in\mathsf{B}\left(  \mathbb{C}^{n,N},\left(  H_{Y}^{-2}\left(
\mathbb{R}^{n}\right)  \right)  ^{n}\right)  &  & \left(  \left(  \gamma
^{\ast}X\right)  \left(  x\right)  \right)  _{j}:=%
{\textstyle\sum\nolimits_{k=1}^{N}}
X_{j,k}\,\delta\left(  x-y^{(k)}\right)  \,,
\end{array}
\label{trace_adj}%
\end{equation}
i,e.: $\mathsf{ran}\left(  \gamma^{\ast}\right)  $ is formed by $H^{-2}$
vector-valued delta distributions supported on $Y$. The singular perturbations
of $L_{0}$ supported on $Y$ are defined by the selfadjoint extensions of the
symmetric restriction $L_{0}\upharpoonright\ker\left(  \gamma\right)  $. These
models are next defined following the approach developped in \cite{MaPo SM}.

\begin{lemma}
\label{Lemma_G_z}Let%
\begin{equation}%
\begin{array}
[c]{ccc}%
G_{z}:=\mathcal{K}_{z}\gamma^{\ast}\,, &  & z\in\mathbb{C}\backslash\left[
0,+\infty\right)  \,.
\end{array}
\label{G_z_def}%
\end{equation}
$i)$ The map $z\rightarrow G_{z}$, $z\in\mathbb{C}\backslash\left[
0,+\infty\right)  $, is analytic $\mathsf{B}\left(  \mathbb{C}^{n,N},\left(
L_{\eta}^{2}\left(  \mathbb{R}^{n}\right)  \right)  ^{n}\right)  $-valued for
all $\eta\in\mathbb{R}$ and, for $\eta>1/2$, continuosly extends to%
\begin{equation}%
\begin{array}
[c]{ccc}%
G_{z}^{\pm}:=\lim_{\varepsilon\rightarrow0^{+}}G_{z\pm i\varepsilon}%
\in\mathsf{B}\left(  \mathbb{C}^{n,N},\left(  L_{-\eta}^{2}\left(
\mathbb{R}^{n}\right)  \right)  ^{n}\right)  \,, &  & \eta>1/2\,.
\end{array}
\label{G_z_map_lim}%
\end{equation}
$ii)$ There exists $c>0$ (possibly depending on $z$ and $\eta$) such that%
\begin{equation}%
\begin{array}
[c]{cc}%
\left\Vert G_{z}X\right\Vert _{\left(  L_{\eta}^{2}\left(  \mathbb{R}%
^{n}\right)  \right)  ^{n}}>c\left\Vert X\right\Vert _{\mathbb{C}^{n,N}}\,, &
\eta\in\mathbb{R}\,.
\end{array}
\label{G_z_low_bound}%
\end{equation}
$iii)$ For $z,z_{0}\in\mathbb{C}\backslash\left[  0,+\infty\right)  $ and
$\lambda>0$ the identities%
\begin{equation}%
\begin{array}
[c]{ccc}%
G_{z}-G_{z_{0}}=\left(  z_{0}-z\right)  \mathcal{K}_{z_{0}}G_{z}\,, &  &
G_{z}^{\pm}-G_{z_{0}}=\left(  z_{0}-z\right)  \mathcal{K}_{z_{0}}G_{z_{0}%
}^{\pm}\,,
\end{array}
\label{G_z_id}%
\end{equation}
hold and
\be\label{G_z_map_1}%
&& \left(  G_{z}-G_{z_{0}}\right)  \in\mathsf{B}\left(  \mathbb{C}^{n,N},\left(
H_{\eta}^{2}\left(  \mathbb{R}^{n}\right)  \right)  ^{n}\right)  \,,\quad
\eta\in\mathbb{R}\,, \nonumber \\
&& \left(  G_{z}^{\pm}-G_{z_{0}}\right)  \in
\mathsf{B}\left(  \mathbb{C}^{n,N},\left(  H_{-\eta}^{2}\left(\mathbb{R}%
^{n}\right)  \right)  ^{n}\right)  \,,\quad\eta>1/2\,.
\en
$iv)$ For $X\in\mathbb{C}^{n,N}$ it results%
\be\label{G_z_id_1}
&&\left(  G_{z}X\right)  _{\ell}\left(  x\right)  =%
{\textstyle\sum\nolimits_{k=1}^{N}}
{\textstyle\sum\nolimits_{j=1}^{n}}
\left(  \Gamma_{z}\left(  x-y^{(k)}\right)  \right)  _{\ell,j}X_{j,k}\,,\\
&&\left(  G_{z}^{\pm}X\right)  _{\ell}\left(  x\right)  =%
{\textstyle\sum\nolimits_{k=1}^{N}}
{\textstyle\sum\nolimits_{j=1}^{n}}
\left(  \Gamma_{z}^{\pm}\left(  x-y^{(k)}\right)  \right)  _{\ell,j}X_{j,k}\,,\nonumber
\en
where $\Gamma_{z}$ denotes the resolvent Green kernels for the Lam\'{e}
operator. Moreover, the outgoing/ingoing Kupradze radiation condition
(\ref{rad_cond_L})/(\ref{rad_cond_T}) hold for $G_{\omega^{2}}^{\pm}$.

\end{lemma}

\begin{proof}
$i)$ The first point follows from (\ref{trace_adj}) and Theorem
\ref{Thorem_res}. $ii)$ By the surjectivity of the trace $\gamma$,
$G_{z}^{\ast}=\gamma\mathcal{K}_{\bar{z}}$ is surjective; hence by the closed
range theorem $G_{z}$ has closed range and (\ref{G_z_low_bound}) follows from
\cite[Theorem VI.5.2]{Kato}. $iii)$ (\ref{G_z_id}) and (\ref{G_z_map_1})
follows from the first resolvent identity%
\begin{equation}%
\begin{array}
[c]{ccc}%
\mathcal{K}_{z}=\mathcal{K}_{z_{0}}+\left(  z_{0}-z\right)  \mathcal{K}%
_{z}\mathcal{K}_{z_{0}}\,, &  & z,z_{0}\in\mathbb{C}\backslash\left[
0,+\infty\right)  \,,
\end{array}
\label{Lame_Res_id}%
\end{equation}
and the mapping properties of $\mathcal{K}_{z}$. $iv)$ Let $X\in
\mathbb{C}^{n,N}$; by (\ref{trace_adj}), (\ref{G_z_def}) we get%
\begin{align}
&  \left.  \left(  G_{z}X\right)  _{\ell}\left(  x\right)  =\left(
\mathcal{K}_{z}%
{\textstyle\sum\nolimits_{k=1}^{N}}
X_{j,k}\,\delta\left(  \cdot-y^{(k)}\right)  \right)  _{\ell}\left(  x\right)
=%
{\textstyle\sum\nolimits_{k=1}^{N}}
{\textstyle\sum\nolimits_{j=1}^{n}}
\left(  \Gamma_{z}\left(  x-y^{(k)}\right)  \right)  _{\ell,j}X_{j,k}%
\,,\right. \\
& \nonumber\\
&  \left.  \left(  G_{z}^{\pm}X\right)  _{\ell}\left(  x\right)  =\left(
\mathcal{K}_{z}^{\pm}%
{\textstyle\sum\nolimits_{k=1}^{N}}
X_{j,k}\,\delta\left(  \cdot-y^{(k)}\right)  \right)  _{\ell}\left(  x\right)
=%
{\textstyle\sum\nolimits_{k=1}^{N}}
{\textstyle\sum\nolimits_{j=1}^{n}}
\left(  \Gamma_{z}^{\pm}\left(  x-y^{(k)}\right)  \right)  _{\ell,j}%
X_{j,k}\,.\right.
\end{align}
Finally, the outgoing/ingoing Kupradze radiation condition follows from
(\ref{G_z_id_1}) and (\ref{rad_cond_L})-(\ref{rad_cond_T}).
\end{proof}

For $z\in\mathbb{C}\backslash\left[  0,+\infty\right)  $ it results (see
\cite[eq. (2.15)]{MaPo SM})%
\begin{align}
&  \left.  \mathsf{dom}\left(  \left(  L_{0}\upharpoonright\ker\left(
\gamma\right)  \right)  ^{\ast}\right) \right.\nonumber\\
& \left. =\left\{  u\in\left(  L^{2}\left(
\mathbb{R}^{n}\right)  \right)  ^{n}\,,\ u=u_{0}+G_{z}X\,,\ u_{0}\in\left(
H^{2}\left(  \mathbb{R}^{n}\right)  \right)  ^{n}\,,\ X\in\mathbb{C}%
^{n,N}\right\}  \,,\right. \label{Lame_adj_dom}\\
& \nonumber\\
&  \left.  \left(  \left(  L_{0}\upharpoonright\ker\left(  \gamma\right)
\right)  ^{\ast}-z\right)  u=\left(  L_{0}-z\right)  u_{0}\,.\right.
\label{Lame_adj_def}%
\end{align}
This representation holds for any $z\in\mathbb{C}\backslash\left[
0,+\infty\right)  $ and the decomposition provided in (\ref{Lame_adj_dom}) is
unique. The action of $G_{z}$ on $\mathbb{C}^{n,N}$ provides a representation
of the defect spaces $\ker\left(  \left(  L_{0}\upharpoonright\ker\left(
\gamma\right)  \right)  ^{\ast}-z\right)  $. Namely%
\begin{equation}
\ker\left(  \left(  L_{0}\upharpoonright\ker\left(  \gamma\right)  \right)
^{\ast}-z\right)  =\mathsf{ran}\left(  G_{z}\right)  \,. \label{Lame_adj_ker}%
\end{equation}

\begin{description}
\item[Assumption] \label{Ass_1}Assume an open set $\mathbb{C}\backslash
\mathbb{R}\subseteq Z_{\Lambda}\subseteq\mathbb{C}\backslash\left[
0,+\infty\right)  $ and a family $\Lambda_{z}\in\mathsf{B}\left(
\mathbb{C}^{n,N}\right)  $, $z\in Z_{\Lambda}$, such that%
\begin{equation}%
\begin{array}
[c]{ccc}%
i)\ \Lambda_{z}^{\ast}=\Lambda_{\bar{z}} &  & ii)\ \Lambda_{w}-\Lambda
_{z}=\left(  z-w\right)  \Lambda_{w}\left(  G_{\bar{w}}\right)  ^{\ast}%
G_{z}\Lambda_{z}\,.
\end{array}
\label{Lamda_assumptions}%
\end{equation}

\end{description}

Following \cite[Theorem 2.4]{MaPo SM} we have:

\begin{theorem}
\label{Theorem_SA_ext}Let $\Lambda:=\left\{  \Lambda_{z}\,,\ z\in Z_{\Lambda
}\right\}  $ be a family of $\mathsf{B}\left(  \mathbb{C}^{n,N}\right)  $ maps
fulfilling the conditions (\ref{Lamda_assumptions}). Then%
\begin{equation}
\mathcal{K}_{z}^{\Lambda}:=\mathcal{K}_{z}+G_{z}\Lambda_{z}\left(  G_{\bar{z}%
}\right)  ^{\ast}\,,\qquad z\in Z_{\Lambda}\,, \label{Krein}%
\end{equation}
is the resolvent of a selfadjoint extension $L_{\Lambda}$ of $L_{0}%
\upharpoonright\ker\left(  \gamma\right)  $.
\end{theorem}

For each $z\in Z_{\Lambda}$, the identity (\ref{Krein}) yields%
\begin{equation}
\mathsf{dom}\left(  L_{\Lambda}\right)  =\left\{  u=u_{0}+G_{z}\Lambda
_{z}\gamma u_{0}\,,\ u_{0}\in\left(  H^{2}\left(  \mathbb{R}^{n}\right)
\right)  ^{n}\right\}  \,. \label{Lame_Lamda_dom}%
\end{equation}
Since $L_{0}\subset L_{\Lambda}\subset\left(  L_{0}\upharpoonright\ker\left(
\tau\right)  \right)  ^{\ast}$, we identify: $L_{\Lambda}:=\left(
L_{0}\upharpoonright\ker\left(  \tau\right)  \right)  ^{\ast}\upharpoonright
\mathsf{dom}\left(  A_{\Lambda}\right)  $; hence%
\begin{equation}
L_{\Lambda}u=\left(  -\mu\Delta-\left(  \lambda+\mu\right)  \nabla
\operatorname{div}\right)  u\quad\text{in }\mathbb{R}^{n}\backslash Y
\label{Lame_Lamda_def}%
\end{equation}
and from (\ref{Lame_adj_ker}) follows%
\begin{equation}
L_{\Lambda}u=L_{0}u_{0}-zG_{z}\Lambda_{z}\gamma u_{0}\,,\qquad u=u_{0}%
+G_{z}\Lambda_{z}\gamma u_{0}\,,\ u_{0}\in\left(  H^{2}\left(  \mathbb{R}%
^{n}\right)  \right)  ^{n}\,. \label{Lame_Lamda_id}%
\end{equation}

\subsubsection{Boundary conditions models}

\label{BCM}

For $\theta\in\mathsf{B}\left(  \mathbb{C}^{n,N}\right)  $ we introduce%
\begin{equation}
\Lambda_{z}\left(  \theta\right)  :=\left(  \theta+\gamma\left(  G_{-1}%
-G_{z}\right)  \right)  ^{-1}\,. \label{Weyl}%
\end{equation}

\begin{lemma}
\label{Lemma_Weyl_fun}Let $\Lambda_{z}\left(  \theta\right)  $ be given by
(\ref{Weyl}) with $\theta\in\mathsf{B}\left(  \mathbb{C}^{n,N}\right)  $
selfadjoint. Then there exists a (possibliy empty) discrete set $S_{\Lambda
\left(  \theta\right)  }\subset\mathbb{R}_{-}$ such that $z\rightarrow
\Lambda_{z}\left(  \theta\right)  $ defines an analytic family of
$\mathsf{B}\left(  \mathbb{C}^{n,N}\right)  $-tensors in $\mathbb{C}%
\backslash\left\{  \left[  0,+\infty\right)  \cup S_{\Lambda\left(
\theta\right)  }\right\}  $. For each $z\in\mathbb{C}\backslash\left\{
\left[  0,+\infty\right)  \cup S_{\Lambda\left(  \theta\right)  }\right\}  $
the Assumption \ref{Ass_1} holds. The limits $\Lambda_{\omega^{2}}^{\pm
}\left(  \theta\right)  :=\lim_{\varepsilon\rightarrow0^{+}}\Lambda
_{\omega^{2}\pm i\varepsilon}\left(  \theta\right)  $ exist in $\mathsf{B}%
\left(  \mathbb{C}^{n,N}\right)  $ for a.a. $\omega>0$, with the possible
exception of a discrete subset and coincide with%
\begin{equation}
\Lambda_{\omega^{2}}^{\pm}\left(  \theta\right)  =\left(  \theta+\gamma\left(
G_{-1}-G_{\omega^{2}}^{\pm}\right)  \right)  ^{-1}\,. \label{Lamda_theta_lim}%
\end{equation}

\end{lemma}

\begin{proof}
Let consider the direct mapping: $\left(  \Lambda_{z}\left(  \theta\right)
\right)  ^{-1}:=\left(  \theta+\gamma\left(  G_{-1}-G_{z}\right)  \right)  $;
by $(i)-(iii)$ of Lemma \ref{Lemma_G_z} results $\left(  \Lambda_{z}\left(
\theta\right)  \right)  ^{-1}\in\mathsf{B}\left(  \mathbb{C}^{n,N}\right)  $
and $\left(  \left(  \Lambda_{z}\left(  \theta\right)  \right)  ^{-1}\right)
^{\ast}=\left(  \Lambda_{\bar{z}}\left(  \theta\right)  \right)  ^{-1}$. By%
\begin{equation}
\left\langle Y,\left(  \Lambda_{z}\left(  \theta\right)  \right)
^{-1}X\right\rangle _{\mathbb{C}^{n,N}}=\left\langle \left(  \Lambda_{\bar{z}%
}\left(  \theta\right)  \right)  ^{-1}Y,X\right\rangle _{\mathbb{C}^{n,N}}\,,
\end{equation}
follows: $\ker\left(  \Lambda_{\bar{z}}\left(  \theta\right)  \right)
^{-1}=\left(  \mathsf{ran}\left(  \left(  \Lambda_{z}\left(  \theta\right)
\right)  ^{-1}\right)  \right)  ^{\bot}$; from (\ref{G_z_low_bound}) follows%
\begin{equation}
\left\vert \left\langle X,\left(  \Lambda_{\bar{z}}\left(  \theta\right)
\right)  ^{-1}X\right\rangle _{\mathbb{C}^{n,N}}\right\vert \geq\left\vert
\operatorname{Im}\left\langle X,\left(  \Lambda_{\bar{z}}\left(
\theta\right)  \right)  ^{-1}X\right\rangle _{\mathbb{C}^{n,N}}\right\vert
=2\left\vert \operatorname{Im}z\right\vert \left\Vert G_{z}X\right\Vert
_{\left(  L^{2}\left(  \mathbb{R}^{n}\right)  \right)  ^{n}}^{2}>2c\left\vert
\operatorname{Im}z\right\vert \,\left\Vert X\right\Vert _{\mathbb{C}^{n,N}}\,.
\end{equation}
Then, $\ker\left(  \Lambda_{\bar{z}}\left(  \theta\right)  \right)
^{-1}=\left\{  0\right\}  $ and $\left(  \Lambda_{z}\left(  \theta\right)
\right)  ^{-1}$ is bijective in $\mathbb{C}\backslash\mathbb{R}$: this allows
to define $\Lambda_{z}\left(  \theta\right)  \in\mathsf{B}\left(
\mathbb{C}^{n,N}\right)  $ for each $z\in\mathbb{C}\backslash\mathbb{R}$;
since $z\rightarrow\left(  \Lambda_{z}\left(  \theta\right)  \right)  ^{-1}$
is analytic (tensor-valued) in $\mathbb{C}\backslash\left[  0,+\infty\right)
$, by the properties of analytic functions in finite dimensional spaces the
inverse tensor-valued map $z\rightarrow\left(  \Lambda_{z}\left(
\theta\right)  \right)  $ is analytic in $\mathbb{C}\backslash\mathbb{R}$ and
extends to a meromorphic function in $\mathbb{C}\backslash\left[
0,+\infty\right)  $. Furthermore, it continuosly extends to the limits
$z\rightarrow\omega^{2}\pm i0$ for a.a. $k>0$ and from%
\be
\mathbf{I}_{n\times N}&=&\left(  \lim_{\varepsilon\rightarrow0^{+}}%
\Lambda_{\omega^{2}\pm i\varepsilon}\left(  \theta\right)  \right)  \left(
\theta+\gamma\left(  G_{-1}-G_{\omega^{2}}^{\pm}\right)  \right)  \nonumber\\
&=& \left(
\theta+\gamma\left(  G_{-1}-G_{\omega^{2}}^{\pm}\right)  \right)  \left(
\lim_{\varepsilon\rightarrow0^{+}}\Lambda_{\omega^{2}\pm i\varepsilon}\left(
\theta\right)  \right)  \,,
\en
the identity (\ref{Lamda_theta_lim}) follows. Let $z\in\mathbb{C}%
\backslash\left\{  \left[  0,+\infty\right)  \cup D_{\Lambda\left(
\theta\right)  }\right\}  $; from: $\left(  \theta+\gamma\left(
G_{-1}-G_{\bar{z}}\right)  \right)  =\left(  \theta+\gamma\left(  G_{-1}%
-G_{z}\right)  \right)  ^{\ast}$ we \ have%
\begin{equation}
\Lambda_{\bar{z}}\left(  \theta\right)  =\left(  \theta+\gamma\left(
G_{-1}-G_{\bar{z}}\right)  \right)  ^{-1}=\left(  \left(  \theta+\gamma\left(
G_{-1}-G_{z}\right)  \right)  ^{\ast}\right)  ^{-1}=\Lambda_{z}^{\ast}\left(
\theta\right)  \,.
\end{equation}
Finally, from (\ref{G_z_id}) follows%
\begin{equation}
\left(  \theta+\gamma\left(  G_{-1}-G_{w}\right)  \right)  -\left(
\theta+\gamma\left(  G_{-1}-G_{z}\right)  \right)  =\gamma\left(  G_{z}%
-G_{w}\right)  =\left(  w-z\right)  G_{\bar{w}}^{\ast}G_{z}\,,
\end{equation}
which implies%
\begin{equation}
\Lambda_{w}\left(  \theta\right)  -\Lambda_{z}\left(  \theta\right)  =\left(
z-w\right)  \Lambda_{w}\left(  \theta\right)  G_{\bar{w}}^{\ast}G_{z}%
\Lambda_{z}\left(  \theta\right)  \,.
\end{equation}

\end{proof}

The construction of Theorem \ref{Theorem_SA_ext} is next implemented with the
family of tensors given in (\ref{Weyl}).

\begin{lemma}
\label{Lemma_Lame_Lambda_dom}Let $\Lambda_{z}\left(  \theta\right)  $ be given
by (\ref{Weyl}) with $\theta\in\mathsf{B}\left(  \mathbb{C}^{n,N}\right)  $
selfadjoint. For $z\in\mathbb{C}\backslash\left\{  \left[  0,+\infty\right)
\cup D_{\Lambda\left(  \theta\right)  }\right\}  $, the domain
(\ref{Lame_Lamda_dom}) rephrases as%
\begin{equation}
\mathsf{dom}\left(  L_{\Lambda\left(  \theta\right)  }\right)  =\left\{
u=u_{0}+G_{-1}X\,,\ u_{0}\in\left(  H^{2}\left(  \mathbb{R}^{n}\right)
\right)  ^{n}\,,\ X\in\mathbb{C}^{n,N}:\gamma u_{0}=\theta X\right\}  \,.
\label{Lame_Lamda_dom_1}%
\end{equation}

\end{lemma}

\begin{proof}
By (\ref{Lame_Lamda_dom}), $u\in\mathsf{dom}\left(  L_{\Lambda\left(
\theta\right)  }\right)  $ implies: $u=u_{0}+G_{z}\Lambda_{z}\gamma u_{0}$,
with $u_{0}\in\left(  H^{2}\left(  \mathbb{R}^{n}\right)  \right)  ^{n}$. By
(\ref{G_z_map_1}) we have%
\begin{equation}%
\begin{array}
[c]{ccc}%
u=\tilde{u}_{0}+G_{-1}\Lambda_{z}\left(  \theta\right)  \gamma u_{0}\,, &  &
\tilde{u}_{0}=u_{0}-\left(  G_{-1}-G_{z}\right)  \Lambda_{z}\left(
\theta\right)  \gamma u_{0}\in\left(  H^{2}\left(  \mathbb{R}^{n}\right)
\right)  ^{n}\,,
\end{array}
\end{equation}
with%
\begin{equation}
\gamma\tilde{u}_{0}=\gamma u_{0}-\gamma\left(  G_{-1}-G_{z}\right)
\Lambda_{z}\left(  \theta\right)  \gamma u_{0}=\gamma u_{0}+\theta\Lambda
_{z}\left(  \theta\right)  \gamma u_{0}-\left(  \Lambda_{z}\left(
\theta\right)  \right)  ^{-1}\Lambda_{z}\left(  \theta\right)  \gamma
u_{0}=\theta\Lambda_{z}\left(  \theta\right)  \gamma u_{0}\,.
\end{equation}
Setting $X=\Lambda_{z}\left(  \theta\right)  \gamma u_{0}$, we get:
$\gamma\tilde{u}_{0}=\theta X$. Hence $u$ belongs to the set
(\ref{Lame_Lamda_dom_1}). Let $u=u_{0}+G_{-1}X\,$, be defined with $u_{0}%
\in\left(  H^{2}\left(  \mathbb{R}^{n}\right)  \right)  ^{n}$ and
$X\in\mathbb{C}^{n,N}$ such that: $\gamma u_{0}=\theta X$. Then%
\begin{equation}%
\begin{array}
[c]{ccc}%
u=\tilde{u}_{0}+G_{z}X\,, &  & \tilde{u}_{0}=u_{0}+\left(  G_{-1}%
-G_{z}\right)  X\in\left(  H^{2}\left(  \mathbb{R}^{n}\right)  \right)
^{n}\,.
\end{array}
\end{equation}
Setting: $\tilde{X}=\left(  \Lambda_{z}\left(  \theta\right)  \right)  ^{-1}X$
we get $u=\tilde{u}_{0}+G_{-1}\Lambda_{z}\left(  \theta\right)  \tilde{X}$
with%
\begin{equation}
\tilde{X}=\left(  \Lambda_{z}\left(  \theta\right)  \right)  ^{-1}X=\theta
X+\gamma\left(  G_{-1}-G_{z}\right)  X=\gamma u_{0}+\gamma\left(  G_{-1}%
-G_{z}\right)  X=\gamma\tilde{u}_{0}\,.
\end{equation}

\end{proof}

With the notation introduced in Sections \ref{Sec_Green_2D} and
\ref{Sec_Green_3D}, $\Gamma_{z}\left(  x-y\right)  $ denotes the resolvent
Green kernel for the Lam\'{e} operator and $\Gamma_{0}\left(  x-y\right)  $
its $z\rightarrow0$ limit: these are $\mathbb{C}^{n,n}$-valued tensors field
and, according to (\ref{Gamma_lim_inv_2D}) and (\ref{Gamma_lim_inv_3D}), the
inverse matrix $\Gamma_{0}^{-1}\left(  x\right)  $ pointwise exists for
$x\neq0$. Let us define the maps $\tau_{j=1,2}:\mathsf{dom}\left(  \left(
L_{0}\upharpoonright\ker\left(  \gamma\right)  \right)  ^{\ast}\right)
\rightarrow\mathbb{C}^{n,N}$%
\begin{align}
&  \left.  \left(  \tau_{1}u\right)  _{j,k}:=\lim_{x\rightarrow y^{(k)}}%
{\textstyle\sum\nolimits_{\ell=1}^{n}}
\left(  \Gamma_{0}^{-1}\left(  x-y^{(k)}\right)  \right)  _{j,\ell}u_{\ell
}\left(  x\right)  \,,\right. \label{tau_1}\\
&  \left.  \left(  \tau_{2}u\right)  _{j,k}:=\lim_{x\rightarrow y^{(k)}%
}\left(  u_{j}\left(  x\right)  -%
{\textstyle\sum\nolimits_{\ell=1}^{n}}
\left(  \Gamma_{0}\left(  x-y^{(k)}\right)  \right)  _{j,\ell}\left(  \tau
_{1}u\right)  _{\ell,k}\right)  \qquad\right.  \label{tau_2}%
\end{align}
for all $j=1,..n\,,\ k=1,..N$. Let us remark that $\tau_{j=1,2}$ extend to
$\left(  H_{\eta}^{2}\left(  \mathbb{R}^{n}\right)  \right)  ^{n}$ functions
for any $\eta\in\mathbb{R}$. These maps allows to represent the operator's
domain in terms of boundary conditions. For $z\in\mathbb{C}\backslash\left[
0,+\infty\right)  $ we introduce $\Xi_{z}\in\mathsf{B}\left(  \mathbb{C}%
^{n,N}\right)  $%
\begin{equation}%
\begin{array}
[c]{ccc}%
\left(  \Xi_{z}\,X\right)  _{j,k}:=%
{\textstyle\sum\nolimits_{k^{\prime}=1\,,\ k^{\prime}\neq k}^{N}}
\left(  \Gamma_{z}\left(  y^{(k)}-y^{(k^{\prime})}\right)  X\right)
_{j,k}\,, &  & X \in\mathbb{C}^{n,N}\,.
\end{array}
\label{Xi_z}%
\end{equation}

\begin{proposition}
\label{Proposition_Lame_dom}Let $\Lambda_{z}\left(  \theta\right)  $ be given
by (\ref{Weyl}) with $\theta\in\mathsf{B}\left(  \mathbb{C}^{n,N}\right)  $
selfadjoint. The domain (\ref{Lame_Lamda_dom}) rephrases as%
\begin{equation}
\mathsf{dom}\left(  L_{\Lambda\left(  \theta\right)  }\right)  =\left\{
u\in\mathsf{dom}\left(  \left(  L_{0}\upharpoonright\ker\left(  \gamma\right)
\right)  ^{\ast}\right)  :\tau_{2}u=\left(  \Xi_{-1}+\theta+\chi
_{-1}\,\mathbf{I}_{n\times N}\right)  \tau_{1}u\right\}  \,,
\label{Lame_Lamda_dom_BC}%
\end{equation}
where $\Xi_{-1}$ and $\chi_{-1}$ are defined by (\ref{Xi_z}) and
(\ref{Gamma_tilda_lim_2D}), (\ref{Gamma_tilda_lim_3D}) for $z=-1$ .

\end{proposition}

\begin{proof}
From (\ref{Gamma_lim_inv_2D_est}) and (\ref{Gamma_lim_inv_3D_est}), follows%
\begin{equation}%
\begin{array}
[c]{ccccc}%
\tau_{1}u_{0}=0\,, &  & u_{0}\in\left(  H_{\eta}^{2}\left(  \mathbb{R}%
^{n}\right)  \right)  ^{n}\,, &  & \eta\in\mathbb{R}\,.
\end{array}
\label{tau_1_id}%
\end{equation}
Moreover%
\begin{align}
&  \left.  \left(  \tau_{1}G_{z}X\right)  _{j,k}=\lim_{x\rightarrow y^{(k)}}%
{\textstyle\sum\nolimits_{\ell=1}^{n}}
\left(  \Gamma_{0}^{-1}\left(  x-y^{(k)}\right)  \right)  _{j,\ell}\left(
{\textstyle\sum\nolimits_{k^{\prime}=1}^{N}}
{\textstyle\sum\nolimits_{j^{\prime}=1}^{n}}
\left(  \Gamma_{z}\left(  x-y^{(k^{\prime})}\right)  \right)  _{\ell
,j^{\prime}}X_{j^{\prime},k^{\prime}}\right)  \right. \nonumber\\
& \nonumber\\
&  \left.  =\lim_{x\rightarrow y^{(k)}}%
{\textstyle\sum\nolimits_{k^{\prime}=1}^{N}}
{\textstyle\sum\nolimits_{j^{\prime}=1}^{n}}
\left(
{\textstyle\sum\nolimits_{\ell=1}^{n}}
\left(  \Gamma_{0}^{-1}\left(  x-y^{(k)}\right)  \right)  _{j,\ell}\left(
\Gamma_{z}\left(  x-y^{(k^{\prime})}\right)  \right)  _{\ell,j^{\prime}%
}\right)  X_{j^{\prime},k^{\prime}}\right. \nonumber\\
& \nonumber\\
&  \left.  =\lim_{x\rightarrow y^{(k)}}%
{\textstyle\sum\nolimits_{k^{\prime}=1}^{N}}
{\textstyle\sum\nolimits_{j^{\prime}=1}^{n}}
\left(  \Gamma_{0}^{-1}\left(  x-y^{(k)}\right)  \Gamma_{z}\left(
x-y^{(k^{\prime})}\right)  \right)  _{j,j^{\prime}}X_{j^{\prime},k^{\prime}%
}\,.\right.
\end{align}
From (\ref{Gamma_lim_inv_2D_est_2}) and (\ref{Gamma_lim_inv_3D_est_2}) we
obtain%
\begin{align}
&  \left.  \left(  \tau_{1}G_{z}X\right)  _{j,k}=\lim_{x\rightarrow y^{(k)}}%
{\textstyle\sum\nolimits_{j^{\prime}=1}^{n}}
\left(  \Gamma_{0}^{-1}\left(  x-y^{(k)}\right)  \Gamma_{z}\left(
x-y^{(k)}\right)  \right)  _{j,j^{\prime}}X_{j^{\prime},k}\right. \nonumber\\
& \nonumber\\
&  \left.  +\lim_{x\rightarrow y^{(k)}}%
{\textstyle\sum\nolimits_{\substack{k^{\prime}=1\\k^{\prime}\neq k}}^{N}}
{\textstyle\sum\nolimits_{j^{\prime}=1}^{n}}
\left(  \Gamma_{0}^{-1}\left(  x-y^{(k)}\right)  \Gamma_{z}\left(
x-y^{(k^{\prime})}\right)  \right)  _{j,j^{\prime}}X_{j^{\prime},k^{\prime}%
}=X_{j,k}\,.\right.
\end{align}
It follows%
\begin{equation}%
\begin{array}
[c]{ccc}%
\tau_{1}G_{z}=\mathbf{I}_{n\times N}\,, &  & z\in\mathbb{C}\backslash\left[
0,+\infty\right)  \,.
\end{array}
\label{tau_1_id_1}%
\end{equation}
Let $u\in\mathsf{dom}\left(  \left(  L_{0}\upharpoonright\ker\left(
\gamma\right)  \right)  ^{\ast}\right)  $; by (\ref{Lame_adj_dom}) there exist
$u_{0}\in\left(  H^{2}\left(  \mathbb{R}^{n}\right)  \right)  ^{n}$ and
$X\in\mathbb{C}^{n,N}$ such that: $u=u_{0}+G_{-1}X$. By (\ref{tau_1_id}) and
(\ref{tau_1_id_1}) we have%
\begin{equation}%
\begin{array}
[c]{ccc}%
\tau_{1}\left(  u_{0}+G_{-1}X\right)  =X\,, &  & u_{0}\in\left(  H_{\eta}%
^{2}\left(  \mathbb{R}^{n}\right)  \right)  ^{n}\,,\ X\in\mathbb{C}%
^{n,N}\,,\ \eta\in\mathbb{R}\,.
\end{array}
\label{tau_1_id_plus}%
\end{equation}
By (\ref{tau_2}), (\ref{G_z_id_1}),and (\ref{tau_1_id_plus}) we have%
\begin{align}
&  \left(  \tau_{2}\left(  u_{0}+G_{-1}X\right)  \right)  _{j,k}\nonumber\\
&  =\left(  u_{0}\right)  _{j}\left(  x_{k}\right)  +\lim_{x\rightarrow
y^{(k)}}\left(
{\textstyle\sum\nolimits_{k^{\prime}=1}^{N}}
{\textstyle\sum\nolimits_{\ell=1}^{n}}
\left(  \Gamma_{-1}\left(  x-y^{(k^{\prime})}\right)  \right)  _{j.\ell
}X_{\ell,k^{\prime}}-%
{\textstyle\sum\nolimits_{\ell=1}^{n}}
\left(  \Gamma_{0}\left(  x-y^{(k)}\right)  \right)  _{j,\ell}X_{\ell
,k}\right) \nonumber\\
& \nonumber\\
&  \left.  =\left(  u_{0}\right)  _{j}\left(  x_{k}\right)  +%
{\textstyle\sum\nolimits_{\substack{k^{\prime}=1\\k^{\prime}\neq k}}^{N}}
{\textstyle\sum\nolimits_{\ell=1}^{n}}
\left(  \Gamma_{-1}\left(  y^{(k)}-y^{(k^{\prime})}\right)  \right)  _{j,\ell
}X_{\ell,k^{\prime}}\right. \nonumber \\
%
& \tcr{ \quad \left.+\lim_{x\rightarrow y^{(k)}} {\textstyle\sum\nolimits_{\ell=1}^{n}}
\left(  \Gamma_{-1}\left(  x-y^{(k)}\right)  -\Gamma_{0}\left(  x-y^{(k)}%
\right)  \right)  _{j,\ell}X_{\ell,k}\,.\right.}  \label{rel_0}%
\end{align}
By (\ref{Gamma_tilda})-(\ref{Gamma_tilda_lim_2D}) and (\ref{Gamma_tilda_3D}%
)-(\ref{Gamma_tilda_lim_3D}) we have%
\begin{equation}
\lim_{x\rightarrow y^{(k)}}%
{\textstyle\sum\nolimits_{\ell=1}^{n}}
\left(  \lim_{\left\vert x\right\vert \rightarrow0}\tilde{\Gamma}_{-1}\left(
x\right)  \right)  _{j,\ell}X_{\ell,k}=\chi_{-1}\,X_{j,k}\,, \label{rel_1}%
\end{equation}
where $\chi_{-1}$ is given by (\ref{Gamma_tilda_lim_2D}) and
(\ref{Gamma_tilda_lim_3D}) for $z=-1$. From (\ref{Xi_z}) we have%
\begin{equation}%
\begin{array}
[c]{ccc}%
\Xi_{-1}\in\mathsf{B}\left(  \mathbb{C}^{N,n}\right)  \,, &  & \left(
\Xi_{-1}\,X\right)  _{j,k}=%
{\textstyle\sum\nolimits_{k^{\prime}=1\,,\ k^{\prime}\neq k}^{N}}
\left(  \Gamma_{-1}\left(  y^{(k)}-y^{(k^{\prime})}\right)  X\right)
_{j,k}\,.
\end{array}
\label{rel_2}%
\end{equation}
Using (\ref{rel_1})-(\ref{rel_2}) allows to rephrase (\ref{rel_0}) as%
\begin{equation}%
\begin{array}
[c]{ccc}%
\tau_{2}\left(  u_{0}+G_{-1}X\right)  =\gamma u_{0}+\left(  \Xi_{-1}+\chi
_{-1}\,\right)  X\,, &  & u_{0}\in\left(  H_{\eta}^{2}\left(  \mathbb{R}%
^{n}\right)  \right)  ^{n}\,,\ X\in\mathbb{C}^{n,N}\,,\ \eta\in\mathbb{R}\,.
\end{array}
\label{tau_2_id}%
\end{equation}
In wiev of (\ref{Lame_adj_dom}), Lemma \ref{Lemma_Lame_Lambda_dom},
(\ref{tau_1_id_plus}) and (\ref{tau_2_id}), the domain (\ref{Lame_Lamda_dom})
identifies with (\ref{Lame_Lamda_dom_BC}).
\end{proof}

\subsubsection{The diffusion problem}

For $\Lambda_{z}\left(  \theta\right)  $ given by (\ref{Weyl}) with $\theta
\in\mathsf{B}\left(  \mathbb{C}^{n,N}\right)  $ selfadjoint, we introduce an
extended operator $\tilde{L}_{\Lambda\left(  \theta\right)  }:\mathsf{dom}%
\left(  \tilde{L}_{\Lambda\left(  \theta\right)  }\right)  \rightarrow
L_{-\eta}^{2}\left(  \mathbb{R}^{3}\right)  $ defined for $\eta>1/2$ and
$z\in\mathbb{C}\backslash\mathbb{R}$ by%
\begin{equation}
\left\{
\begin{array}
[c]{l}%
\mathsf{dom}\left(  \tilde{L}_{\Lambda\left(  \theta\right)  }\right)
=\left\{  u\in\left(  L_{-\eta}^{2}\left(  \mathbb{R}^{n}\right)  \right)
^{n}\,,\ u=u_{0}+G_{z}\Lambda_{z}\left(  \theta\right)  \gamma u_{0}%
\,,\ u_{0}\in\left(  H_{-\eta}^{2}\left(  \mathbb{R}^{n}\right)  \right)
^{n}\right\}  \,,\\
\\
\tilde{L}_{\Lambda\left(  \theta\right)  }u=\left(  -\mu\Delta-\left(
\lambda+\mu\right)  \nabla\operatorname{div}\right)  u\,,\qquad\text{in
}\mathbb{R}^{n}\backslash Y
\end{array}
\right.  \label{Lame_Lamda_ext}%
\end{equation}
\ This model can be characterized in terms of the boundary conditions
introduced in Proposition \ref{Proposition_Lame_dom}.

\begin{lemma}
Let $\Lambda_{z}\left(  \theta\right)  $ and $\tilde{L}_{\Lambda\left(
\theta\right)  }$ be defined according to (\ref{Weyl}) and
(\ref{Lame_Lamda_ext}) with $\theta\in\mathsf{B}\left(  \mathbb{C}%
^{n,N}\right)  $ selfadjoint. Then, for each $k$ such that the limits
(\ref{Lamda_theta_lim}) exist, $\mathsf{dom}\left(  \tilde{L}_{\Lambda\left(
\theta\right)  }\right)  $ admits the representations%
\begin{equation}
\mathsf{dom}\left(  \tilde{L}_{\Lambda\left(  \theta\right)  }\right)
=\left\{  u\in\left(  L_{-\eta}^{2}\left(  \mathbb{R}^{n}\right)  \right)
^{n}\,,\ u=u_{0}+G_{\omega^{2}}^{\pm}\Lambda_{\omega^{2}}^{\pm}\left(
\theta\right)  \gamma u_{0}\,,\ u_{0}\in\left(  H_{-\eta}^{2}\left(
\mathbb{R}^{n}\right)  \right)  ^{n}\right\}  \,, \label{Lame_Lamda_ext_dom}%
\end{equation}
and the boundary conditions hold%
\begin{equation}%
\begin{array}
[c]{ccc}%
\tau_{2}u=\left(  \Xi_{-1}+\theta+\chi_{-1}\right)  \tau_{1}u\,, &  &
u\in\mathsf{dom}\left(  \tilde{L}_{\Lambda\left(  \theta\right)  }\right)  \,.
\end{array}
\label{Lame_Lamda_ext_dom_BC}%
\end{equation}

\end{lemma}

\begin{proof}
By (\ref{G_z_id})-(\ref{G_z_map_1}), $u=u_{0}+G_{z}\Lambda_{z}\left(
\theta\right)  \gamma u_{0}$ identifies with%
\begin{equation}%
\begin{array}
[c]{ccc}%
u=\tilde{u}_{0}+G_{\omega^{2}}^{\pm}\Lambda_{z}\left(  \theta\right)  \gamma
u_{0}\,, &  & \tilde{u}_{0}=u_{0}+\left(  \omega^{2}-z\right)  \mathcal{K}%
_{\omega^{2}}^{\pm}G_{z}\Lambda_{z}\left(  \theta\right)  \gamma u_{0}%
\in\left(  H_{-\eta}^{2}\left(  \mathbb{R}^{n}\right)  \right)  ^{n}\,.
\end{array}
\end{equation}
Using (\ref{Lamda_assumptions}) we get%
\begin{equation}
\Lambda_{z}\left(  \theta\right)  \gamma u_{0}=\Lambda_{\omega^{2}}^{\pm
}\left(  \theta\right)  \left(  1+\left(  \omega^{2}-z\right)  \left(
G_{\omega^{2}}^{\mp}\right)  ^{\ast}G_{z}\Lambda_{z}\left(  \theta\right)
\right)  \gamma u_{0}=\Lambda_{\omega^{2}}^{\pm}\left(  \theta\right)
\gamma\tilde{u}_{0}\,,
\end{equation}
from which it follows%
\begin{equation}%
\begin{array}
[c]{ccc}%
u=\tilde{u}_{0}+G_{\omega^{2}}^{\pm}\Lambda_{\omega^{2}}^{\pm}\left(
\theta\right)  \tilde{u}_{0}\,, &  & \tilde{u}_{0}=u_{0}+\left(  \omega
^{2}-z\right)  \mathcal{K}_{\omega^{2}}^{\pm}G_{z}\Lambda_{z}\left(
\theta\right)  \gamma u_{0}\in\left(  H_{-\eta}^{2}\left(  \mathbb{R}%
^{n}\right)  \right)  ^{n}\,.
\end{array}
\end{equation}
This shows that%
\begin{equation}
\mathsf{dom}\left(  \tilde{L}_{\Lambda\left(  \theta\right)  }\right)
\subseteq\left\{  u\in\left(  L_{-\eta}^{2}\left(  \mathbb{R}^{n}\right)
\right)  ^{n}\,,\ u=u_{0}+G_{\omega^{2}}^{\pm}\Lambda_{\omega^{2}}^{\pm
}\left(  \theta\right)  \gamma u_{0}\,,\ u_{0}\in\left(  H_{-\eta}^{2}\left(
\mathbb{R}^{n}\right)  \right)  ^{n}\right\}  \,.
\end{equation}
Using again (\ref{G_z_id})-(\ref{G_z_map_1}) and (\ref{Lamda_assumptions}) a
similar argument leads to the opposite inclusion. Proceeding as in Lemma
\ref{Lemma_Lame_Lambda_dom}, we get%
\begin{equation}
\mathsf{dom}\left(  \tilde{L}_{\Lambda\left(  \theta\right)  }\right)
=\left\{  u=u_{0}+G_{-1}X\,,\ u_{0}\in\left(  H_{-\eta}^{2}\left(
\mathbb{R}^{n}\right)  \right)  ^{n}\,,\ X\in\mathbb{C}^{n,N}:\gamma
u_{0}=\theta X\right\}  \,.
\end{equation}
and the boundary conditions (\ref{Lame_Lamda_ext_dom_BC}) follows from
(\ref{tau_1_id_plus}) and (\ref{tau_2_id}).
\end{proof}

For $\omega>0$ such that the limits (\ref{Lamda_theta_lim}) exist, the
generalized eigenfunctions of energy $\omega^{2}$ are the solutions of the
problem%
\begin{equation}
\left(  \tilde{L}_{\Lambda\left(  \theta\right)  }-\omega^{2}\right)
u=0\,,\qquad u\in\mathsf{dom}\left(  \tilde{L}_{\Lambda\left(  \theta\right)
}\right)  \,. \label{gen_eigenfun_eq}%
\end{equation}

\begin{lemma}
\label{Lemma_gen_eigenfun}For $\omega>0$ such that the limits
(\ref{Lamda_theta_lim}) exist, the solutions of (\ref{gen_eigenfun_eq})
express as%
\begin{equation}%
\begin{array}
[c]{ccc}%
u=u_{0}+G_{\omega^{2}}^{\pm}\Lambda_{\omega^{2}}^{\pm}\left(  \theta\right)
\gamma u_{0}\,, &  &
\end{array}
\label{gen_eigenfun_id}%
\end{equation}
where $u_{0}\in\left(  H_{-\eta}^{2}\left(  \mathbb{R}^{n}\right)  \right)
^{n}$ is a generalized eigenfunction of $L_{0}$.

\end{lemma}

\begin{proof}
Let $u\in\left(  H_{-\eta}^{2}\left(  \mathbb{R}^{n}\right)  \right)  ^{n}$ be
a generalized eigenfunction of $L_{0}$, and define $u$ according to
(\ref{gen_eigenfun_id}). Then, by (\ref{Lame_Lamda_ext}), we get%
\begin{equation}
\left(  \tilde{L}_{\Lambda\left(  \theta\right)  }-\omega^{2}\right)  \left(
u_{0}+G_{\omega^{2}}^{\pm}\Lambda_{\omega^{2}}^{\pm}\left(  \theta\right)
\gamma u_{0}\right)  =-\left(  \mu\Delta+\left(  \lambda+\mu\right)
\nabla\operatorname{div}+\omega^{2}\right)  u_{0}=0\,.
\end{equation}

\end{proof}

Let $u^{sc}$ denote the stationary diffusion of an incoming wave
$u^{in}:=u_{0}\in\left(  H_{-\eta}^{2}\left(  \mathbb{R}^{n}\right)  \right)
^{n} $ (a generalized eigenfunction of $L_{0}$); we have: $u=u^{in}+u^{sc}%
\in\mathsf{dom}\left(  \tilde{L}_{\Lambda\left(  \theta\right)  }\right)  $
and by (\ref{Lame_Lamda_ext_dom_BC}), $u^{sc}$ solves the boundary condition
problem%
\begin{equation}
\left\{
\begin{array}
[c]{l}%
\begin{array}
[c]{ccc}%
\left(  \mu\Delta+\left(  \lambda+\mu\right)  \nabla\operatorname{div}%
+\omega^{2}\right)  u^{sc}=0\,, &  & \text{in }\mathbb{R}^{n}\backslash\,Y
\end{array}
\\
\\
\tau_{2}\left(  u^{sc}+u^{in}\right)  =\left(  \Xi_{-1}+\theta+\chi
_{-1}\right)  \tau_{1}\left(  u^{sc}+u^{in}\right)  \,,
\end{array}
\right.  \label{diff_eq}%
\end{equation}
and fulfills the Kupradze outgoing radiation conditions (\ref{Kupradze}).
By Lemma \ref{Lemma_gen_eigenfun}, this problem admits the unique solution%
\begin{equation}
u^{sc}=G_{\omega^{2}}^{+}\Lambda_{\omega^{2}}^{+}\left(  \theta\right)  \gamma
u^{in}\,. \label{diff_id}%
\end{equation}

\begin{remark}
The construction presented above provides a large class of point perturbation
models, including anisotropic and non-local interactions. From the physical
point of view, anisotropy refers to different scattering properties depending
on the direction, while non-locality refers to a coupling between different
point scatterers. While the applications considered in this work focus on
isotropic local perturbations, it is worth noticing that the scattering theory
presented here holds in a much more general framework.
\end{remark}

\subsubsection{Modelling local isotropic point perturbations}

\label{subsec:point}

Let us start with a characterization of the class of models we are interested
in. The notion of isotropy corresponds to the fact that the effect of the
perturbation is independent from the direction, while locality excludes the
possible couplings between the points. These properties and the domain
representation in (\ref{Lame_Lamda_dom_BC}) motivate the next definition.

\begin{defn}
Let $\theta\in\mathsf{B}\left(  \mathbb{C}^{n,N}\right)  $ be selfadjoint and
let $\left\{  e_{j,k}\right\}  $ denote the standard basis in $\mathbb{C}%
^{n,N}$. Let $\Xi_{-1}$ and $\chi_{-1}$ be defined by (\ref{Xi_z}),
(\ref{Gamma_tilda_lim_2D}) and (\ref{Gamma_tilda_lim_3D}) for $z=-1$. We say
that the operator $L_{\Lambda\left(  \theta\right)  }$, with $\Lambda
_{z}\left(  \theta\right)  $ given in (\ref{Weyl}), is an\emph{ isotropic
}perturbation when%
\begin{equation}%
\begin{array}
[c]{ccc}%
\left\langle e_{j,k},\left(  \theta+\Xi_{-1}+\chi_{-1}\,\mathbf{I}_{n\times
N}\right)  \left(  e_{j,k^{\prime}}\right)  \right\rangle _{\mathbb{C}^{n,N}%
}=\left\langle e_{j^{\prime},k},\left(  \theta+\Xi_{-1}+\chi_{-1}%
\,\mathbf{I}_{n\times N}\right)  \left(  e_{j^{\prime},k^{\prime}}\right)
\right\rangle _{\mathbb{C}^{n,N}}\,,
\end{array}
\label{iso_def}%
\end{equation}
\tcr{for all $j,j^{\prime}=1,..n$  and $k,k^{\prime}=1,..N\,.$}
The perturbation is \emph{local} if%
\begin{equation}%
\begin{array}
[c]{ccc}%
\left\langle e_{j,k},\left(  \theta+\Xi_{-1}+\chi_{-1}\,\mathbf{I}_{n\times
N}\right)  \left(  e_{j^{\prime},k^{\prime}}\right)  \right\rangle
_{\mathbb{C}^{n,N}}=0\,, &  & \forall\,j,j^{\prime}=1,..n\text{ and }k\neq
k^{\prime}\,.
\end{array}
\label{local_def}%
\end{equation}

\end{defn}

As an example, consider $\alpha\in\mathbb{R}^{N,N}$ and define $\alpha
\in\mathsf{B}\left(  \mathbb{C}^{n,N}\right)  $ as%
\begin{equation}%
\begin{array}
[c]{ccc}%
\alpha\left(  M\right)  :=M\alpha\,, &  & M\in\mathbb{C}^{n,N}\,.
\end{array}
\label{iso_mod}%
\end{equation}
The operator $\theta\left(  \alpha\right)  :=\alpha-\Xi_{-1}-\chi
_{-1}\,\mathbf{I}_{n\times N}\in\mathsf{B}\left(  \mathbb{C}^{N,n}\right)  $
is selfadjoint. This choice enters in the scheme of the Proposition
\ref{Proposition_Lame_dom} and a corresponding perturbed Lam\'{e} operator
$L_{\Lambda\left(  \theta\left(  \alpha\right)  \right)  }$ is defined with
the boundary conditions (\ref{Lame_Lamda_dom}). Is easy to see in this
framework that the scalar products in (\ref{iso_def}) are independent from the
direction. Indeed, we have
\begin{equation}%
\begin{array}
[c]{ccc}%
\left\langle e_{j,k},\left(  \theta\left(  \alpha\right)  +\Xi_{-1}+\chi
_{-1}\,\mathbf{I}_{n\times N}\right)  \left(  e_{j,k^{\prime}}\right)
\right\rangle _{\mathbb{C}^{n,N}}=\left\langle e_{j,k},\alpha\left(
e_{j,k^{\prime}}\right)  \right\rangle _{\mathbb{C}^{n,N}}=\alpha_{k^{\prime
}k}\,, &  & \forall\,j=1,..n\,.
\end{array}
\end{equation}
Hence the perturbation $L_{\Lambda\left(  \theta\left(  \alpha\right)
\right)  }$ is isotropic. The condition of locality is satisfied when%
\[%
\begin{array}
[c]{ccc}%
\left\langle e_{j,k},\left(  \theta\left(  \alpha\right)  +\Xi_{-1}+\chi
_{-1}\,\mathbf{I}_{n\times N}\right)  \left(  e_{j^{\prime},k^{\prime}%
}\right)  \right\rangle _{\mathbb{C}^{n,N}}=\left\langle e_{j,k},\alpha\left(
e_{j^{\prime},k^{\prime}}\right)  \right\rangle _{\mathbb{C}^{n,N}}%
=\alpha_{k^{\prime}k}=0\,, &  & \forall\,k\neq k^{\prime}\,,
\end{array}
\]
which corresponds to the choice of a diagonal $\alpha$. We resume these
properties in the following lemma.

\begin{lemma}
\label{Lemma_iso_loc}Let $\alpha\in\mathbb{R}^{N}$ and define $\alpha
\in\mathsf{B}\left(  \mathbb{C}^{n,N}\right)  $ as%
\begin{equation}%
\begin{array}
[c]{ccc}%
\alpha\left(  M\right)  :=M%
\begin{pmatrix}
\alpha_{1} &  & \\
& \ddots & \\
&  & \alpha_{N}%
\end{pmatrix}
\in\mathbb{C}^{n,N}\,, &  & M\in\mathbb{C}^{n,N}\,.
\end{array}
\label{alpha}%
\end{equation}
The perturbation $L_{\alpha}:=L_{\Lambda\left(  \theta\left(  \alpha\right)
\right)  }$, defined by $\theta\left(  \alpha\right)  :=\alpha-\Xi_{-1}%
-\chi_{-1}\,\mathbf{I}_{n\times N}$, is isotropic and local.
\end{lemma}

In what follows, we consider an isotropic and local perturbation $L_{\alpha}$
with $\alpha\in\mathsf{B}\left(  \mathbb{C}^{n,N}\right)  $ selfadjoint given
by (\ref{alpha}). The boundary conditions in (see Proposition
\ref{Proposition_Lame_dom}) are%
\begin{equation}%
\begin{array}
[c]{ccc}%
\tau_{2}u=\,\alpha\left(  \tau_{1}u\right)  \,, &  & u\in\mathsf{dom}\left(
\left(  L_{0}\upharpoonright\ker\left(  \gamma\right)  \right)  ^{\ast
}\right)  \,,
\end{array}
\label{alpha_BC}%
\end{equation}
and componentwise reads as%
\begin{equation}%
\begin{array}
[c]{ccc}%
\left(  \tau_{2}u\right)  _{j,k}=\alpha_{k}\,\left(  \tau_{1}u\right)
_{j,k}\,, &  & j=1,..n\,,\ k=1,..N\,.
\end{array}
\label{alpha_BC_1}%
\end{equation}
With the notation introduced in Proposition \ref{Proposition_Lame_dom}, we
define%
\begin{equation}
\Lambda_{z}^{\alpha}:=\left(  \alpha-\Xi_{-1}-\chi_{-1}\,\mathbf{I}_{n\times
N}+\gamma\left(  G_{-1}-G_{z}\right)  \right)  ^{-1}\,. \label{Lamda_alpha_z}%
\end{equation}
Setting%
\begin{equation}
\mathsf{dom}\left(  L_{\alpha}\right)  =\left\{  u\in\mathsf{dom}\left(
\left(  L_{0}\upharpoonright\ker\left(  \gamma\right)  \right)  ^{\ast
}\right)  :\tau_{2}u=\alpha\,\left(  \tau_{1}u\right)  \right\}  \,,
\label{Lame_alpha_dom}%
\end{equation}
and using the construction of Theorem \ref{Theorem_SA_ext} and Proposition
\ref{Proposition_Lame_dom}, we have%
\begin{equation}
L_{\alpha}:=\left(  L_{0}\upharpoonright\ker\left(  \gamma\right)  \right)
^{\ast}\upharpoonright\mathsf{dom}\left(  L_{\alpha}\right)  \,.
\label{Lame_alpha}%
\end{equation}
Rephrasing in this framework the result of the previous sections, by Lemma
\ref{Lemma_Weyl_fun} the limit maps $\Lambda_{\omega^{2}}^{\alpha,\pm}%
\in\mathsf{B}\left(  \mathbb{C}^{n,N}\right)  $ exist for $\omega^{2}%
\in\left(  0,+\infty\right)  \backslash S_{\alpha}$ where $S_{\alpha}%
\subset\left(  0,+\infty\right)  $ is a discrete subset. Under this condition,
the stationary diffusion problem of an incoming wave $u^{in}$ reads as%
\begin{equation}
\left\{
\begin{array}
[c]{l}%
\begin{array}
[c]{ccc}%
\left(  \mu\Delta+\left(  \lambda+\mu\right)  \nabla\operatorname{div}%
+\omega^{2}\right)  u^{sc}=0\,, &  & \text{in }\mathbb{R}^{n}\backslash\,Y
\end{array}
\\
\tau_{2}\left(  u^{sc}+u^{in}\right)  =\alpha\,\tau_{1}\left(  u^{sc}%
+u^{in}\right)  \,,\\
\text{The outgoing radiation conditions in (\ref{Kupradze}) hold.}%
\end{array}
\right.  \label{diff_eq_alpha}%
\end{equation}

\begin{lemma}
\label{lemm} Let $\alpha\in\mathsf{B}\left(  \mathbb{C}^{N,n}\right)  $ be
defined by (\ref{alpha}), $\omega^{2}\in\left(  0,+\infty\right)  \backslash
S_{\alpha}$ and $u^{in}\in\left(  H_{-\eta}^{2}\left(  \mathbb{R}^{n}\right)
\right)  ^{n}$ be a generalized eigenfunction of energy $\omega^{2}$ of
$L_{0}$. The unique solution of (\ref{diff_eq_alpha}) is given by%
\begin{align}
&  \left.  \left(  u^{sc}\right)  _{\ell}\left(  x\right)  =\frac{1}{\mu}%
{\textstyle\sum\nolimits_{k=1}^{N}}
\Phi_{k_{s}}^{+}\left(  x-y^{(k)}\right)  \left(  \Lambda_{\omega^{2}}%
^{\alpha,+}\gamma u^{in}\right)  _{\ell,k}\left(  x\right)  \right.
\nonumber\\
&  \left.  +\frac{1}{\omega^{2}}%
{\textstyle\sum\nolimits_{k=1}^{N}}
\partial_{\ell}\left(
{\textstyle\sum\nolimits_{j=1}^{n}}
\partial_{j}\left(  \Phi_{k_{s}}^{+}\left(  x-y^{(k)}\right)  -\Phi_{k_{p}%
}^{+}\left(  x-y^{(k)}\right)  \right)  \left(  \Lambda_{\omega^{2}}%
^{\alpha,+}\gamma u^{in}\right)  _{j,k}\left(  x\right)  \right)  \,.\right.
\label{diff_id_alpha}%
\end{align}

\end{lemma}

\begin{proof}
The representation (\ref{diff_id_alpha}) follows from (\ref{diff_id}) by
taking into account (\ref{Kupradze_res_lim}) and (\ref{G_z_id_1}).
\end{proof}

\begin{remark}
From the identity%
\begin{align}
&  \left.  \gamma\left(  G_{-1}-G_{z}\right)  =\lim_{x\rightarrow0}\left(
\Gamma_{-1}\left(  x\right)  -\Gamma_{z}\left(  x\right)  \right)
\mathbf{I}_{n}+\Xi_{-1}-\Xi_{z}\right. \nonumber\\
&  \left.  =\lim_{x\rightarrow0}\left(  \Gamma_{-1}\left(  x\right)
-\Gamma_{0}\left(  x\right)  \right)  \mathbf{I}_{n}-\lim_{x\rightarrow
0}\left(  \Gamma_{z}\left(  x\right)  -\Gamma_{z}\left(  x\right)  \right)
\mathbf{I}_{n}+\Xi_{-1}-\Xi_{z}\,,\right. \nonumber
\end{align}
and the limits (\ref{Gamma_tilda_lim_2D}), (\ref{Gamma_tilda_lim_3D}), follows%
\begin{equation}
\gamma\left(  G_{-1}-G_{z}\right)  =\chi_{-1}\mathbf{I}_{n\times N}-\chi
_{z}\,\mathbf{I}_{n\times N}+\Xi_{-1}-\Xi_{z}\,.
\end{equation}
Then (\ref{Lamda_alpha_z}) rephrases as%
\begin{equation}
\Lambda_{z}^{\alpha}:=\left(  \alpha-\chi_{z}\,\mathbf{I}_{n\times N}-\Xi
_{z}\right)  ^{-1}\,,
\end{equation}
and the corresponding limits $\Lambda_{\omega^{2}}^{\alpha,+}$ corresponds to
the inverse of%
\begin{align}
&  \left.  \left(  \Lambda_{\omega^{2}}^{\alpha}\right)  ^{-1}:=\left(
\alpha-\chi_{\omega^{2}}\mathbf{I}_{n\times N}-\Xi_{\omega^{2}}\right)
\right. \nonumber\\
& \nonumber\\
&  \left.  =%
\begin{pmatrix}
\left(  \alpha_{1}-\chi_{\omega^{2}}\right)  \mathbf{I}_{n} & -\Gamma
_{\omega^{2}}\left(  y^{(1)}-y^{(2)}\right)  & \cdots & -\Gamma_{\omega^{2}%
}\left(  y^{(1)}-y^{(N)}\right) \\
-\Gamma_{\omega^{2}}\left(  y^{(2)}-y^{(1)}\right)  & \left(  \alpha_{2}%
-\chi_{\omega^{2}}\right)  \mathbf{I}_{n} & \cdots & -\Gamma_{\omega^{2}%
}\left(  y^{(2)}-y^{(N)}\right) \\
\vdots &  &  & \vdots\\
-\Gamma_{\omega^{2}}\left(  y^{(N)}-y^{(1)}\right)  & \cdots & -\Gamma
_{\omega^{2}}\left(  y^{(N)}-y^{(N-1)}\right)  & \left(  \alpha_{n}%
-\chi_{\omega^{2}}\right)  \mathbf{I}_{n}%
\end{pmatrix}
\,.\right.  \label{Lamda_alpha_lim}%
\end{align}
The representation (\ref{Lamda_alpha_lim}) is consistent with the ones
provided in \cite[Sec. II and III]{HuSi} for isotropic point perturbations
models using the regularization approach.
\end{remark}

\bigskip

\section{Elastic scattering by point-like and extended obstacles}

\label{sec:multiscale} \setcounter{equation}{0}

In this section, we consider the scattering of elastic incident waves from a
multi-scale scatterer $\Omega=D\cup Y$, where $D$ is an extended obstacle and
$Y:=\{y^{(j)}:j=1,2,\cdots,N\}\subset{\mathbb{R}}^{n}\backslash\tcr{\overline
{D}}$ represents a set of finitely many point-like elastic scatterers. For
simplicity we assume that the extended scatterer $D$ is a rigid elastic body.
However, our augments can be easily adpted to other penetrable or impenetrable
extended scatterers.

In what follows we focus on the case where $\Omega$ is formed by point
scatterer and an extended sound-soft (i.e. Dirichlet) obstacle. We will next
denote with $u(x)=u^{in}(x)+u^{sc}(x)$ the total field corresponding to the
multiple scattering of an incident wave $u^{in}$ on the point scatterers in
$Y$ and the Dirichlet extended obstacle $D$, while $u_{D}=u^{in}+u_{D}^{sc}%
\in(H_{loc}^{1}({\mathbb{R}}^{n}\backslash\overline{\Omega}))^{n}$ is the
total field in the absence of the point obstacles, i.e.: $u_{D}^{sc}$ is the
unique Kupradze outgoing radiation-solution to the boundary value problem
\begin{equation}
\left(  \Delta^{\ast}+\omega^{2}\right)  u_{D}^{sc}=0\quad\mbox{in}\quad
{\mathbb{R}}^{n}\backslash\overline{D},\qquad u_{D}^{sc}=-u^{in}%
\quad\mbox{on}\quad\partial D\,.\label{BVP:point}%
\end{equation}
The fundamental solution for the Navier equation in ${\mathbb{R}}%
^{n}\backslash\overline{D}$ with Dirichlet boundary condition on $\partial D$,
next denoted with $\Gamma_{D}(x,y)$, is a $\mathbb{C}^{n,n}$ tensor field
defined by: $\Gamma_{D}(x,y):=\Gamma_{D}^{sc}(x,y)+\Gamma_{\omega^{2}}(x,y)$
where $\Gamma_{\omega^{2}}(x,y)$ ($y\in{\mathbb{R}}^{n}\backslash\overline{D}%
$) is the free space Green's tensor to the Navier equation, while $\Gamma
_{D}^{sc}(x,y)$ is the unique solution to
\begin{equation}
\left(  \Delta^{\ast}+\omega^{2}\right)  \Gamma_{D}^{sc}(\cdot,y)=0\quad
\mbox{in}\quad{\mathbb{R}}^{n}\backslash\overline{D},\qquad\Gamma_{D}%
^{sc}(\cdot,y)=-\Gamma_{\omega^{2}}(\cdot,y)\quad\mbox{on}\quad\partial D\,.
\end{equation}

Motivated by the "impedance"-type boundary condition (\ref{diff_eq_alpha}) for
modelling local and isotropic point perturbations established in Section
\ref{sec:point-interaction}, we assume that the boundary conditions
(\ref{alpha_BC_1}) hold for the total field, i.e.%
\begin{equation}%
\begin{array}
[c]{ccccc}%
(\tau_{2}u)_{j,k}=(\tau_{1}u)_{j,k}\alpha_{k}\,, &  & \alpha\in{\mathbb{C}%
}^{N}{\mathbb{\,}}, &  & j=1,..n\,,\ k=1,..N\,.
\end{array}
\label{aps3}%
\end{equation}
According to the definition of the mappings $\tau_{\ell=1,2}$ (see
(\ref{tau_1}) and (\ref{tau_2})) we have the following asymptotic behavior%
\begin{equation}
u_{j}(x)=%
{\textstyle\sum_{j^{\prime}=1,..n}}
\left(  \Gamma_{0}(x,y^{(k)})\right)  _{j,j^{\prime}}\,(\tau_{1}u)_{j^{\prime
},k}+(\tau_{2}u)_{j,k}+O(|x-y^{(k)}|)\qquad\mbox{as}\quad x\rightarrow
y^{(k)}.\label{aps4}%
\end{equation}
Let $\alpha=(\alpha_{1},\alpha_{2},\cdots,\alpha_{N})\in\mathbb{C}^{N}$ and
define as before $\alpha\in\mathsf{B}\left(  \mathbb{C}^{n,N}\right)  $
according to (\ref{alpha}); to describe the solution of (\ref{Navier}),
(\ref{Kupradze}) and (\ref{aps3}), we introduce the modified tensor
$\Lambda_{\omega^{2}}^{\alpha,D}\in\mathsf{B}\left(  \mathbb{C}^{n,N}\right)
$, whose inverse is defined by the $\mathbb{C}^{n,n}$ matrix-block entries%
\begin{equation}\label{MatrixTheta}
\begin{array}
[c]{ccc}%
\left(  \Lambda_{\omega^{2}}^{\alpha,D}\right)  _{k,k^{\prime}}^{-1}=\left\{
\begin{array}
[c]{lll}%
-\Gamma_{D}(y^{(k)},y^{(k^{\prime})})\,, &  & k\neq k^{\prime}\,,\\
\left(  \alpha_{k}-\chi_{\omega^{2}}\right)  \mathbf{I}_{n}\,, &  &
k=k^{\prime}\,,
\end{array}
\right.   &  & k,k^{\prime}=1,..N\,,
\end{array}
\end{equation}
with the constant $\chi_{\omega^{2}}$ given by (\ref{Gamma_tilda_lim_2D}) and
(\ref{Gamma_tilda_lim_3D}) for $z=\omega^{2}$. We define the set%
\begin{equation}
S_{\alpha}^{D}:=\{\omega>0:\mbox{det}\left(  \Lambda_{\omega^{2}}^{\alpha
,D}\right)  ^{-1}=0\}\,,\label{S_D_alpha}%
\end{equation}
and address the exterior stationary diffusion problem of an incoming wave
$u^{in}$%
\begin{equation}
\left\{
\begin{array}
[c]{l}%
\begin{array}
[c]{ccc}%
\left(  \Delta^{\ast}+\omega^{2}\right)  u^{sc}=0\,, &  & \text{in }%
\mathbb{R}^{n}\backslash\,\Omega\,,
\end{array}
\\%
\begin{array}
[c]{ccc}%
u^{sc}=-u^{in}\,, &  & \text{on }\partial D\,,
\end{array}
\\
\tau_{2}\left(  u^{sc}+u^{in}\right)  =\alpha\,\tau_{1}\left(  u^{sc}%
+u^{in}\right)  \,,\\
\text{The outgoing radiation conditions in (\ref{Kupradze}) hold.}%
\end{array}
\right.  \label{diff_eq_alpha_D}%
\end{equation}
We remark that, when $\alpha\in\mathbb{R}^{N}$, this multiscale scattering
model can be justified either using the Foldy's formal approach, the renormalization technique or the
point-interaction approach, following the same arguments presented in Section
\ref{sec:point-interaction}. In particular, the solution can be derived as in
Lemmas \ref{Lemma_gen_eigenfun} and \ref{lemm} \tcr{by replacing the Green's tensor
$\Gamma_{\omega^{2}}$ with $\Gamma_{D}$ and
by replacing the incident wave $u^{in}$ with
$u_{D}^{in}$}, respectively (see (\ref{PointSolution}) below).


\begin{theorem}
\label{Th1}Assume that $\omega\notin S_{\alpha}^{D}$ and $\mbox{Im}\,\alpha
_{j}\leq0$ for all $k=1,\cdots,N$. Then, the boundary value problem problem
(\ref{diff_eq_alpha_D}) admits a unique solution in $H_{loc}^{1}({\mathbb{R}%
}^{3}\backslash\overline{\Omega})$, which represents as%
\begin{equation}%
\begin{array}
[c]{ccc}%
u(x)=u_{D}(x)+\sum\limits_{k,k^{\prime}=1,..N}\Gamma_{D}(x,y^{(k)})\,\left(
\Lambda_{\omega^{2}}^{\alpha,D}\right)  _{k,k^{\prime}}\,u_{D}\left(
y^{\left(  k^{\prime}\right)  }\right)  \,, &  & x\in{\mathbb{R}}%
^{3}\backslash\overline{\Omega}\,,
\end{array}
\label{PointSolution}%
\end{equation}
where $u_{D}=u^{in}+u_{D}^{sc}$ is the total field in the absence of the
point-like obstacles and the sums over the space-idices are hidden.
\end{theorem}

\tcr{The above theorem shows that
the
scattered field caused by $D\cup Y$ consists of two parts: one is due to the diffusion by the extended scatterer (i.e., $u_D$) and
the other one is a linear combination of the interactions between the point-like obstacles and
the interaction between the point-like obstacles with the extended one (i.e., those terms appearing in the summation).}
We next present a
more direct proof to check this point; the result is slightly more general,
since we only assume now a sign condition for $\operatorname{Im}\alpha$.
\begin{proof}
We carry out the proof in 3D only, since the 2D case can be treated analogously.

(i) Uniqueness. Assuming $u^{in}=0$, we only need to prove that $u^{sc}=0$.
Note that $u^{sc}=0$ on $\partial D$ and $u^{sc}$ fulfills the conditions
(\ref{aps3}) as well as the Kupradze's radiation condition.

To prove $u^{sc}\equiv0$, we need the analogue of Rellich's lemma in linear
elasticity (see e.g., \cite[Lemma 5.8]{Hahner} and \cite{BHSY2018}) The
Rellich's lemma for the Helmholtz equation can be found in \cite[Chapter
2]{CK}, which ensures uniqueness for solutions to exterior boundary value
problems. For $a,b\in R$ such that $a+b=\lambda+\mu$, define the sesquilinear
form $\mathcal{E}_{a,b}$ and the traction operator $T_{a,b}$ via
\begin{align*}
\mathcal{E}_{a,b}\,(u,v) &  :=(a+\mu)\sum_{j,k=1}^{3}\frac{\partial u_{j}%
}{\partial x_{k}}\frac{\partial v_{j}}{\partial x_{k}}+b\,(\nabla\cdot
u)(\nabla\cdot v)-a\,\mathrm{curl\,}u\cdot\mathrm{curl\,}v,\\
T_{a,b}\,u &  :=(a+\mu)\frac{\partial u}{\partial\nu}+b\,\mathrm{div\,}%
u\,\nu+a\,\nu\times\mathrm{curl\,}v,
\end{align*}
where $u=(u_{1},u_{2},u_{3})$ and $v=(v_{1},v_{2},v_{3})$. In the generalized
Betti's formula (see \cite{Kupradze}), we take a special choice of the
parameters $a=-\mu$ and $b=\lambda+2\mu$, so that $a+b=\lambda+\mu$. For
notational convenience we write $T=T_{-\mu,\lambda+2\mu}$ and $\mathcal{E}%
=\mathcal{E}_{-\mu,\lambda+2\mu}$ to indicate the dependance of the traction
operator $T$ and the sesquilinear form $\mathcal{E}$ on these parameters.
Choose $\epsilon>0$ sufficiently small and $R>0$ sufficiently large such that
\[
D\subset B_{R},\;B_{\epsilon}(y^{(j)})\subset B_{R}\backslash\overline
{D},\quad B_{\epsilon}(y^{(j)})\cap B_{\epsilon}(y^{(m)})=\emptyset
\]
for all $j,m=1,2,\cdots,N$ and $j\neq m$. Applying the generalized Betti's
formula (see \cite{Kupradze}) for $u^{sc}$ to the region $B_{R,\epsilon}%
=B_{R}\backslash\overline{D}\backslash\{\cup_{j=1}^{N}\overline{B_{\epsilon
}(y^{(j)})}\}$, we find
\begin{align}
0 &  =-\int_{B_{R,\epsilon}}(\Delta u^{sc}+\omega^{2}\,u^{sc})\,\overline
{u^{sc}}\,dx\nonumber\label{2-1-1}\\
&  =\int_{B_{R,\epsilon}}\mathcal{E}(u^{sc},\overline{u^{sc}})dx-\int%
_{\partial B_{R,\epsilon}}Tu^{sc}\cdot\overline{u^{sc}}\,ds\nonumber\\
&  =\int_{B_{R,\epsilon}}\mathcal{E}(u^{sc},\overline{u^{sc}})dx-\int%
_{|x|=R}Tu^{sc}\cdot\overline{u^{sc}}\,ds+\sum_{j=1}^{N}\int_{\partial
B_{\epsilon}(y^{(j)})}Tu^{sc}\cdot\overline{u^{sc}}\,ds,
\end{align}
where the normal directions at $\partial B_{\epsilon}(y^{(j)})$ or $\partial
B_{R,\epsilon}$ are assumed to point outward. Here we have used the vanishing
of $u^{sc}$ on $\partial D$. Next we estimate the integral on $\partial
B_{\epsilon}(y^{(j)})$ in (\ref{2-1-1}) by using the impedance-type boundary
condition (\ref{aps3}). Setting $C_{j}:=(\tau_{1}u^{sc})_{j}\in{\mathbb{C}%
}^{3}$, we derive from (\ref{aps4}) and (\ref{Gamma_Res_ker_3D_lim}) that
\[
u^{sc}(x)=\frac{\lambda+3\mu}{8\pi\mu(\lambda+2\mu)}\frac{C_{j}}{|x-y^{(j)}%
|}+\frac{\lambda+3\mu}{8\pi\mu(\lambda+\mu)}\frac{(x-y^{(j)})\otimes
(x-y^{(j)})}{|x-y^{(j)}|^{3}}\cdot C_{j}+\alpha_{j}C_{j}+o(1)
\]
as $x\rightarrow y^{(j)}$. Let $F(x)=(x\otimes x)\cdot C_{j}/|x|^{3}$.
Straightforward calculations show that
\[
\mathrm{div\,}F(x)=\frac{C_{j}\cdot\hat{x}}{|x|^{2}},\quad\mathrm{curl\,}%
F(x)=\frac{C_{j}\times\hat{x}}{|x|^{2}}%
\]
where $\hat{x}=x/|x|\in{\mathbb{S}}$. Making use the previous relation, one
can calculate for $x\in\partial B_{\epsilon}(y^{(j)})$ that
\begin{align*}
&  (\lambda+2\mu)\mathrm{div\,}u^{sc}\,\nu=-\frac{\nu(\nu\cdot C_{j})}%
{4\pi\epsilon^{2}}+O(1),\\
&  \nu\times\mathrm{curl\,}u^{sc}=-\frac{\nu\times(C_{j}\times\nu)}%
{4\pi\epsilon^{2}}+O(1),
\end{align*}
as $\epsilon\rightarrow0$, where $\nu(x)=(x-y^{(j)})/\epsilon$ on $\partial
B_{\epsilon}(y^{(j)})$. By definition of the traction operator, it then
follows that
\[
Tu^{sc}\cdot\overline{u}^{sc}=-\frac{-1}{16\pi^{2}\epsilon^{3}\mu}|C_{j}%
|^{2}-\frac{\overline{\alpha_{j}}}{4\pi\epsilon^{2}}|C_{j}|^{2}+O(\frac
{1}{\epsilon})\quad\mbox{on}\quad\partial B_{\epsilon}(y^{(j)})
\]
as $\epsilon\rightarrow0$. Since $\mathrm{Im\,}\alpha_{j}\leq0$, we get via
the mean value theorem that
\[
\lim_{\epsilon\rightarrow0}\;\mathrm{Im\,}\left(  \int_{\partial B_{\epsilon
}(y^{(j)})}Tu^{sc}\cdot\overline{u^{sc}}\,ds\right)  =\frac{\mathrm{Im\,}%
\alpha_{j}}{4\pi\epsilon^{2}}|C_{j}|^{2}\leq0.
\]
Now, taking the imaginary part of (\ref{2-1-1}) and letting $\epsilon$ tend to
zero yield
\[
\mathrm{Im\,}\left(  \int_{\Gamma_{R}}Tu^{sc}\cdot\overline{u^{sc}%
}\,ds\right)  \leq0.
\]
Applying \cite[Lemma 5.8]{Hahner} gives $u^{sc}=0$, which proves uniqueness.

(ii) Existence. The solution (\ref{PointSolution}) obviously satisfies the
Navier equation in ${\mathbb{R}}^{3}\backslash\overline{\Omega}$ and the
Dirichlet boundary condition on $\partial D$. Moreover, $u^{sc}=u-u^{in}$
fulfills the Kupradze radiation condition, because both $u_{D}^{sc}$ and
$\Gamma_{D}$ are radiating solutions. Hence, it suffices to check the
Impedance-type boundary condition in (\ref{aps3}) imposed on $y^{(k)}$,
$k=1,\cdots,N$. From the definition of $\tau_{1}$ and $\tau_{2}$, follows
\[
(\tau_{2}\Gamma^{D}(x,z))_{k}-\alpha_{k}\,(\tau_{1}\Gamma^{D}(x,z))_{k}%
=\left\{
\begin{array}
[c]{lll}%
\Gamma^{D}(y^{(j)},z), &  & \mbox{if}\;z\in{\mathbb{R}}^{3}\backslash Y\,,\\
\Gamma^{D}(y^{(j)},y^{(m)}), &  & \mbox{if}\;z=y^{(m)}\in Y\,,\ m\neq k\,,\\
\chi_{\omega^{2}}-\alpha_{k}, &  & \mbox{if}\;z=y^{(k)}\in Y\,,
\end{array}
\right.
\]
where, here and in the following, we hide the space-indexes and the
corresponding sums. This leads to
\[
(\tau_{2}\Gamma(x,y^{(m)}))_{k}-\alpha_{k}\,(\tau_{1}\Gamma(x,y^{(m)}%
))_{k}=-\left(  \Lambda_{\omega^{2}}^{\alpha,D}\right)  _{m,k}^{-1}\,,\qquad
m,k=1,\cdots,N.
\]
On the other hand, since $u^{in}$ is of $C^{\infty}$-smoothness at $y^{(k)}$,
it holds that
\[
(\tau_{1}u^{in})_{k}=0,\quad(\tau_{1}u^{in})_{k}=u^{in}(y^{(k)}).
\]
Consequently, by direct calculation we have for $k=1,\cdots,N$ that
\begin{align*}
&  \quad(\tau_{2}u)_{k}-\alpha_{k}\,(\tau_{1}u)_{k}\\
&  =(\tau_{2}u^{in})_{k}-\alpha_{k}(\tau_{1}u^{in})_{k}+(\tau_{2}u^{sc}%
)_{k}-\alpha_{k}(\tau_{1}u^{sc})_{k}\\
&  =u^{in}(y^{(k)})+%
{\textstyle\sum_{m,l=1}^{N}}
\left(  \Lambda_{\omega^{2}}^{\alpha,D}\right)  _{m,l}\left[  (\tau_{2}%
(\Gamma(\cdot,y^{\left(  m\right)  }))_{k}-\alpha_{k}(\tau_{1}(\Gamma
(\cdot,y^{\left(  m\right)  }))_{k}\right]  \,u^{in}(y^{(l)})\\
&  =u^{in}(y^{(k)})-%
{\textstyle\sum_{l=1}^{N}}
\left[
{\textstyle\sum_{m=1}^{N}}
\left(  \Lambda_{\omega^{2}}^{\alpha,D}\right)  _{l,m}\left(  \Lambda
_{\omega^{2}}^{\alpha,D}\right)  _{m,k}^{-1}\right]  u^{in}(y^{(l)})\\
&  =u^{in}(y^{(k)})-u^{in}(y^{(k)})=0\,.
\end{align*}

\end{proof}

\section{Inverse problems}

\label{inverse} \setcounter{equation}{0}

\tcr{Having established the forward scattering model in Theorem \ref{Th1},
we consider in this section the inverse problem of  simultaneously recovering the shape of the extended elastic body $D$ and the location of point-like scatterers $y^{(j)}\in \R^n\backslash\overline{D}$, $j=1,2,\cdots, N$, in the
case of isotropic and local interactions. The number $N$ of the point-like scatterers is always assumed to be finite but unknown.
The factorization method \cite{Kirsch1998, K2008} by Kirsch will be adapted to
such a multi-scale inverse scattering problem by using different type of
elastic waves.
}

\subsection{Factorization method}

We consider three inverse problems at a fixed frequency as follows:

\begin{description}
\item[(i)] Recover the shape $\partial\Omega$ of the extended obstacle and the
location $\{y^{(j)}: j=1,2,\cdots,N\}$ of point-like scatterers from knowledge
of the entire far-field pattern over all incident and observation directions,
that is, $\{u^{\infty}(\hat{x},d): \hat{x},d\in\mathbb{S}^{n-1}\}$.

\item[(ii)] Recover $\partial\Omega$ and $\{y^{(j)}: j=1,2,\cdots,N\}$ from
P-part of the far-field pattern over all observation directions excited by
incident compressional plane waves with all incident directions.

\item[(iii)] Recover $\partial\Omega$ and $\{y^{(j)}: j=1,2,\cdots,N\}$ from
S-part of the far-field pattern over all observation directions excited by
incident shear plane waves with all incident directions.
\end{description}

Evidently, only one type of elastic waves is used in the last two inverse
problems, whereas the entire wave is involved in the first problem. Before
stating the factorization method we need to define the far-field operator in
linear elasticity. For $g(d)\in L^{2}(\mathbb{S}^{n-1})^{n}$, $d\in
{\mathbb{S}}^{n-1}$, we have the decomposition $g(d)=g_{s}(d)+g_{p}(d)$ where
\begin{align}
\label{g}g_{p}(d):=(g(d)\cdot d)\; d,\quad g_{s}(d):=\left\{
\begin{array}
[c]{lll}%
d\times g(d)\times d \quad\mbox{if}\quad n=3; &  & \\
(g(d)\cdot d^{\bot})\; d^{\bot}\;\;\mbox{if}\quad n=2. &  &
\end{array}
\right.
\end{align}
Obviously, $g_{s}$ belongs to the space of transversal vector fields on
$\mathbb{S}^{n-1}$ defined as
\begin{align}
\label{gs}L^{2}_{s}(\mathbb{S}^{n-1}):=\{g_{s}: \mathbb{S}^{N-1}%
\rightarrow{\mathbb{C}}^{n}: g_{s}(d)\cdot d=0,\; |g_{s}|\in L^{2}%
(\mathbb{S}^{n-1})\}
\end{align}
while $g_{p}$ belongs to the space of longitudinal vector fields on
$\mathbb{S}^{n-1}$:
\begin{align*}
L^{2}_{p}(\mathbb{S}^{n-1}):=\{g_{p}: \mathbb{S}^{n-1}\rightarrow{\mathbb{C}%
}^{n}: g_{p}(d)\times d=0\;\mbox{in}\;{\mathbb{R}}^{3},\, g_{p}(d)\cdot
d^{\bot}=0\;\mbox{in}\;{\mathbb{R}}^{2},\, |g_{p}|\in L^{2}(\mathbb{S}%
^{n-1})\}.
\end{align*}
For $g\in L^{2}(\mathbb{S}^{n-1})^{n}$, introduce the incident fields
\begin{align*}
v^{in}_{g}(x)  &  :=\int_{{\mathbb{S}}^{n-1}} \left[  g_{s}(d)e^{ik_{s}x\cdot
d}+g_{p}(d) e^{ik_{p}x\cdot d} \right]  \,d s(d),\\
v^{in}_{g_{s}}(x)  &  :=\int_{{\mathbb{S}}^{n-1}} \left[  g_{s}(d)e^{ik_{s}%
x\cdot d} \right]  \,d s(d),\\
v^{in}_{g_{p}}(x)  &  :=\int_{{\mathbb{S}}^{n-1}} \left[  g_{p}(d)
e^{ik_{p}x\cdot d} \right]  \,d s(d).
\end{align*}

\begin{defn}
Let $v_{g}^{\infty}$ be the far-field pattern of the incident wave $v^{in}%
_{g}$, and let $v_{g,p}^{\infty}$ (resp. $v_{g,s}^{\infty}$ ) be the
longitudinal (resp. transversal) far-field pattern of the incident wave
$v^{in}_{g,p}$ ($v^{in}_{g,p}$). The far-field operators $F$, $F_{p}$ and
$F_{s}$ are defined by
\begin{align*}
Fg  &  :=v_{g}^{\infty}\,,\quad(L^{2}(\mathbb{S}^{n-1}))^{n}\rightarrow
(L^{2}(\mathbb{S}^{n-1}))^{n},\\
F_{s} g_{s}  &  :=v_{g,s}^{\infty}\,,\quad(L^{2}_{s}(\mathbb{S}^{n-1}%
))^{n}\rightarrow(L^{2}_{s}(\mathbb{S}^{n-1}))^{n},\\
F_{p} g_{p}  &  :=v_{g,p}^{\infty}\,,\quad(L^{2}_{p}(\mathbb{S}^{n-1}%
))^{n}\rightarrow(L^{2}_{p}(\mathbb{S}^{n-1}))^{n}.
\end{align*}

\end{defn}

Our inverse scattering problems (IP1) and (IP2) can be equivalently stated as
the problems of finding $\partial\Omega:=\partial D\cup Y$ from the far-field
operators $F$, $F_{p}$ and $F_{s}$. The operators $F_{p}$ and $F_{s}$ are
related to $F$ as follows:
\begin{align*}
F_{p}=\Pi_{p}\, F\,\Pi_{p}^{*},\quad F_{s}=\Pi_{s}\, F\,\Pi_{s}^{*},
\end{align*}
where $\Pi_{\alpha}$: $(L^{2}(\mathbb{S}^{n-1}))^{n}\rightarrow(L_{\alpha}%
^{2}(\mathbb{S}^{n-1}))^{n}$( $\alpha=p,s$) is the orthogonal projector
operator defined by
\begin{align*}
\Pi_{p} g(d)=g_{p}(d),\quad\Pi_{s} g(d)=g_{s}(d), \qquad g\in(L^{2}%
(\mathbb{S}^{n-1}))^{n}.
\end{align*}
The adjoint operator $\Pi_{\alpha}^{*}: (L_{\alpha}^{2}(\mathbb{S}^{n-1}%
))^{n}\rightarrow(L^{2}(\mathbb{S}^{n-1}))^{n}$ of $\Pi_{\alpha}$ is just the
inclusion operator from $(L_{\alpha}^{2}(\mathbb{S}^{n-1}))^{n}$ to
$(L^{2}(\mathbb{S}^{n-1}))^{n}$. To state the factorization scheme, for
$\mathbf{a}\in\mathbb{S}^{n-1}$ and $y\in{\mathbb{R}}^{n}$ we denote by
$\Pi^{\infty}_{\mathbf{a},y}$ the far-field pattern of the outgoing function
$x\rightarrow\Pi(x,y)\mathbf{a}$. Its compressional and shear parts will be
denoted by $\Pi^{\infty,p}_{\mathbf{a},y}$ and $\Pi^{\infty,s}_{\mathbf{a},y}%
$, respectively. Define $\mathcal{F}_{\#}:=|\mathrm{Re\,} \mathcal{F} | +
|\mathrm{Im\,} \mathcal{F}|$ for $\mathcal{F}=F$, $F_{p}$ and $F_{s}$, where
$\mathrm{Re\,} \mathcal{F}:=(F+F^{*})/2$ and $\mathrm{Im\,} \mathcal{F}%
:=(F-F^{*})/(2i)$.

Given a non-vanishing vector $\mathbf{a}\in\mathbb{S}^{n-1}$ and a sampling
point $y\in{\mathbb{R}}^{n}$, the far-field pattern of the radiation function
$x\rightarrow\Gamma_{\omega}(x, y)\mathbf{a}$ is given by
\begin{align*}
\Gamma_{\mathbf{a},y}^{\infty}(\hat{x})= \left\{
\begin{array}
[c]{lll}%
e^{-ik_{s} \hat{x}\cdot y}[\hat{x}\times(\mathbf{a}\times{\hat{x}%
})]+e^{-ik_{p} \hat{x}\cdot y}(\hat{x}\cdot\mathbf{a})\hat{x}\quad &  &
\mbox{if}\quad n=3,\\
e^{-ik_{s} \hat{x}\cdot y} (\mathbf{a}\cdot{\hat{x}}^{\bot})\hat{x}^{\bot
}+e^{-ik_{p} \hat{x}\cdot y}(\hat{x}\cdot\mathbf{a})\hat{x}\quad &  &
\mbox{if}\quad n=2,
\end{array}
\right.
\end{align*}
where $\hat{x}^{\bot}\in{\mathbb{S}}^{n-1}$ is orthogonal to $\hat{x}$. By the
definition of the projection operators $\Pi_{p}$ and $\Pi_{s}$,
\begin{align*}
\Pi_{s}[\Gamma_{\mathbf{a},y}^{\infty}(\hat{x})]  &  = \left\{
\begin{array}
[c]{lll}%
e^{-ik_{s} \hat{x}\cdot y}[\hat{x}\times(\mathbf{a}\times{\hat{x}})]\quad &  &
\mbox{if}\quad n=3,\\
e^{-ik_{s} \hat{x}\cdot y} (\mathbf{a}\cdot{\hat{x}}^{\bot})\hat{x}^{\bot
}\quad &  & \mbox{if}\quad n=2,
\end{array}
\right. \\
\Pi_{p}[\Gamma_{\mathbf{a},y}^{\infty}(\hat{x})]  &  = e^{-ik_{p} \hat{x}\cdot
y}(\hat{x}\cdot\mathbf{a})\hat{x}\qquad\qquad\qquad\mbox{if}\quad n=2,3.
\end{align*}
Following \cite{HMS,HKS13}, one can prove that the function $\Gamma
_{\mathbf{a},y}^{\infty}(\hat{x})$ (resp. $\Pi_{\alpha}[\Gamma_{\mathbf{a}%
,y}^{\infty}(\hat{x})]$, $\alpha=p,s$) belongs to the range of $\mathcal{F}%
_{\#}$ (resp. $\mathcal{F}_{\alpha\#}$) if and only if $y\in D\cup Y$. By
Picard's theorem (see, e.g., \cite[Theorem 4.8]{K2008}, the set $D\cup Y$ can
be characterized through the eigensystems of the far-field operators $F$,
$F_{s}$ and $F_{p}$ as follows.

\begin{theorem}
\label{TH:factorization} Suppose that $\omega^{2}$ is not the Dirichlet
eigenvalue of $-\Delta^{*}$ in $D$, $\omega\notin S_{\alpha}$ and
$\mbox{Im}\,\alpha_{j}\le0$ for all $j=1,2,\cdots,N$. Then $y\in\Omega=D\cup
Y$ if and only if one of the following three criterions holds
\begin{align}
\label{Chap}W(y)  &  :=\left[  \sum_{n=1}^{\infty}\frac{|(g_{n},
\Gamma^{\infty}_{\mathbf{a},y})_{L^{2}(\mathbb{S}^{n-1})}|^{2}}{\zeta_{n}%
}\right]  ^{-1} >0,\\
W_{\alpha}(y)  &  :=\left[  \sum_{n=1}^{\infty}\frac{|(g_{n}^{(\alpha)},
\Pi_{\alpha}[\Gamma^{\infty}_{\mathbf{a},y}])_{L^{2}(\mathbb{S}^{n-1})}|^{2}%
}{\zeta_{n}^{(\alpha)}}\right]  ^{-1} >0,\qquad\alpha=p,s
\end{align}
where $\{\zeta_{n},g_{n}\}$ (resp. $\{\zeta_{n}^{(\alpha)},g^{(\alpha)}_{n}\}$
) is an eigensystem of the operator $F_{\#}$ (resp. $F_{\alpha\#}$).
\end{theorem}

Theorem \ref{TH:factorization} can be proved by combining the arguments of
\cite{HMS} for inverse acoustic scattering from multi-scale sound-soft
scatterers and \cite{HKS13} where the factorization method using different
type of elastic waves was established for imaging an extended rigid body (that
is, $Y=\emptyset$). We remark that, based on the solution representation
(\ref{PointSolution}), the far-field operators $F$, $F_{s}$, $F_{p}$ can be
factorized as follows (cf. \cite{HMS, HKS13}):
\begin{align*}
F=GS^{*}G^{*}, \quad F_{\alpha}=(\Pi_{\alpha}G) S^{*} (\Pi_{\alpha}%
G)^{*},\quad\alpha=p,s,
\end{align*}
where $G$ and $S$ are modified data-pattern and single layer operators. There
is no essential difficulties in applying the range identity established in
\cite{K2008}. The assumptions on the frequency $\omega$ ensure that the
operators $F$, $F_{s}$ and $F_{p}$ are injective with dense range and also
give arise to properties of the middle operator required by the range identity.




\subsection{Numerical tests in 2D}

\FloatBarrier In this section, we present numerical examples in two dimensions
for testing accuracy and validity of the developed inversion schemes. In our
numerical tests, the solutions $u_{D}^{sc}$ and $\Gamma_{D}^{sc}$ are
generated by means of the boundary integral equation method and we always set
$\omega=8$, $\rho=1$, $\lambda=2$, $\mu=1$. For simplicity, we assume that
$\omega^{2}$ is not the Dirichlet eigenvalue of $-\Delta^{\ast}$ and make the
ansatz for the scattered fields corresponding to the rigid scatterer $D$ in
the form
\begin{align*}
u_{D}^{sc}(x;d) &  =\int_{\partial D}\Gamma_{\omega^{2}}(x,y)\varphi
(y;d)ds(y),\\
\Gamma_{D}^{sc}(x;z,\mathbf{a}) &  =\int_{\partial D}\Gamma_{\omega^{2}%
}(x,y)\psi(y;z,\mathbf{a})ds(y),
\end{align*}
where $u_{D}^{sc}(x;d)$ and $\Gamma_{D}^{sc}(x;z,\mathbf{a})$ are excited by a
plane wave of direction $d\in{\mathbb{S}}^{1}$ and the point source located at
$z\in{\mathbb{R}}^{2}\backslash\overline{D}$ with the polarization vector
$\mathbf{a}\in{\mathbb{S}}$, respectively. The densities functions
$\varphi(y;d)$ and $\psi(y;z,\mathbf{a})$ can be solved by an equivalent
boundary integral equation of first kind. Then, for $z=y^{(m)}\in Y$ one can
get the far-field pattern of $\Gamma_{D}^{sc}(x;z,\mathbf{a})$ as
\[
\Gamma_{D,\mathbf{a}}^{\infty}(\hat{x};y^{(m)})=\int_{\partial D}\left[
e^{-ik_{s}\hat{x}\cdot y}(\psi(y;y^{(m)},\mathbf{a})\cdot{\hat{x}}^{\bot}%
)\hat{x}^{\bot}+e^{-ik_{p}\hat{x}\cdot y}(\hat{x}\cdot\psi(y;y^{(m)}%
,\mathbf{a}))\hat{x}\right]  \;ds(y).
\]
By (\ref{PointSolution}), we get the far-field pattern for scattering of a
plane wave from $\Omega=D\cup Y$ as
\[
u^{\infty}(\hat{x};d)=\sum_{m,l=1}^{N}\Gamma_{D}^{\infty}(\hat{x}%
;y^{(m)})\,\left[  \Lambda_{\omega^{2}}^{\alpha,D}\right]  _{m,l}\,u_{D}%
(y^{(l)};d),\quad\hat{x},d\in{\mathbb{S}},
\]
where $\Gamma_{D}^{\infty}=(\Gamma_{D,\mathbf{e}_{1}}^{\infty},\Gamma
_{D,\mathbf{e}_{2}}^{\infty})\in{\mathbb{R}}^{2\times2}$ with $\mathbf{e}%
_{1}=(0,1)$, $\mathbf{e}_{2}=(1,0)$.

Let $N$ (here we set $N=64$) incident compressional plane waves $u_{p}%
^{in}=d_{j}\exp(ik_{p}x\cdot d_{j})$ or shear plane waves $u_{s}^{in}%
=d_{j}^{\perp}\exp(ik_{s}x\cdot d_{j})$ be given at equidistantly distributed
directions $d_{j}=(\cos\theta_{j},\sin\theta_{j})$ with $\theta_{j}=2\pi j/N$,
$j=1,2,\cdots,N$. Denote by $u_{p}^{\infty}(\hat{x};d_{j})$, $u_{s}^{\infty
}(\hat{x};d_{j})$ the P-part, S-part of the far-field pattern of the scattered
field $u^{sc}$ corresponding to $\Omega$ and the incident compressional wave,
and by $u_{p}^{\infty}(\hat{x};d_{j}^{\perp})$, $u_{s}^{\infty}(\hat{x}%
;d_{j}^{\perp})$ the counterparts associated with the incident shear wave.
\begin{figure}[th]
\centering
\includegraphics[scale=0.3]{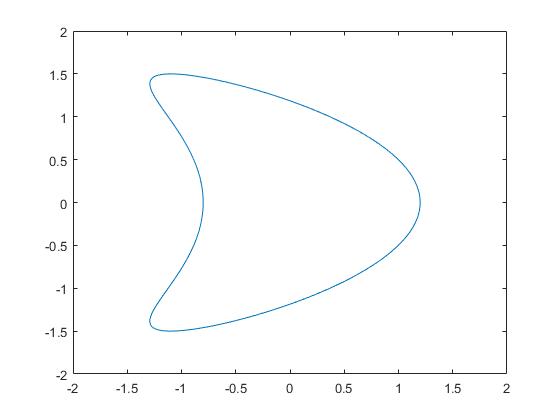}\caption{The kite-shaped
extended obstacle.}%
\label{obstacle}%
\end{figure}The numerical experiments are performed in the following three cases.

\begin{description}
\item[PP case:] Reconstruct $\partial D$ and $Y$ from $u_{p}^{\infty}(d_{k};
d_{j})$ for $N$ incident compressional plane waves $d_{j}\exp(ik_{p}x\cdot
d_{j})$.

\item[SS case:] Reconstruct $\partial D$ and $Y$ from $u_{s}^{\infty}(d_{k};
d_{j}^{\perp})$ for $N$ incident shear plane waves $d_{j}^{\perp}\exp
(ik_{s}x\cdot d_{j}) $.

\item[FF case:] Reconstruct $\partial D$ and $Y$ from $u^{\infty}(d_{k}%
;d_{j})d_{k}+u^{\infty}(d_{k};d_{j}^{\perp})d_{k}^{\perp}$ for $N$ incident
elastic plane waves $d_{j}\exp(ik_{p}x\cdot d_{j})+d_{j}^{\perp}\exp
(ik_{s}x\cdot d_{j})$.
\end{description}

The measured far-field and near-field data are perturbed by the multiplication
of $(1+\delta\xi)$ with the noise level $\delta$, where $\xi$ is an
independent and uniformly distributed random variable generated between -1 and
1. In all of our examples, we suppose that $D$ is a kite-shaped extended
obstacle; see Figure \ref{obstacle}.

\textbf{Example 1:} (point-scatterers on a line segment) In this example, we
compare the reconstruction results for the kite-shaped extended obstacle
together with $M=6$ point-like scatterers on the line segment $2\times[-2,2]$
in the SS case, PP case and FF case; see Figures \ref{E1.1} and \ref{E1.2}.
Here we set $\alpha_{j}=0.1, j=1,\cdots,M$. It can be seen that using the
S-part of the far-field pattern still produces satisfactory reconstruction,
but the reconstruction from P-part of the far-field pattern is less reliable.
This is due the fact that shorter wave length always gives higher resolution
and the wave length of shear waves is smaller than that of compressional
waves, that is, $\lambda_{s}<\lambda_{p}$; see (\ref{wavelength}). In our
tests we have $\lambda_{s}\sim0.785$ and $\lambda_{s}\sim1.57$. In Figure
\ref{E1.1}, the distance between the point-like scatterers is $l=0.8$, which
is close to the shear wave length but is nearly a half of the compressional
wave length. Hence, the point-like scatterers can be distinguished in the SS
case rather than the PP case. However, when $l$ is decreased to be 0.4 in
Figure \ref{E1.2}, using shear waves cannot yield satisfactory reconstructions
of the point-like scatterers; cf. (c) and (d) in Figures \ref{E1.1} and
\ref{E1.2}. It can also be observed that using entire elastic waves would give
more reliable reconstructions than the PP and SS cases.

\textbf{Example 2:} (sensitivity to the ``impedance'' coefficients) In this
example, we consider the FF case for reconstructing the same scatterers as in
Example 1. We fix $M=6$, $\mathbf{a}=(1,0)$. The reconstruction results for
different values of $\alpha_{j}$ are presented in Figure \ref{E2.1}. It can be
observed that the values of the indicator function around the point-like
scatterers grow as the value of $\alpha_{j}$ decreases, that is, the
point-like obstacles are more visible for small $\alpha_{j}$. This can be
interpreted by the fact that the far-field pattern is proportional to the
inverse of $\alpha_{j}$; see the solution form (\ref{PointSolution}) and the
diagonal terms appearing in the matrix (\ref{MatrixTheta}). The same
phenomenon was observed in the inverse acoustic scattering problems \cite{HMS}.

\textbf{Example 3:} (random distributed point-like scatterers) In this
example, we consider the kite-shaped obstacle and 20 point-like scatterers
randomly located in $\{[-3,-2]\cup[2,3]\}\times[-3,3]$. The reconstruction
from data with 1\% noise are shown in Figure \ref{E3.1}.

\begin{figure}[th]
\centering
\begin{tabular}
[c]{cc}%
\includegraphics[scale=0.2]{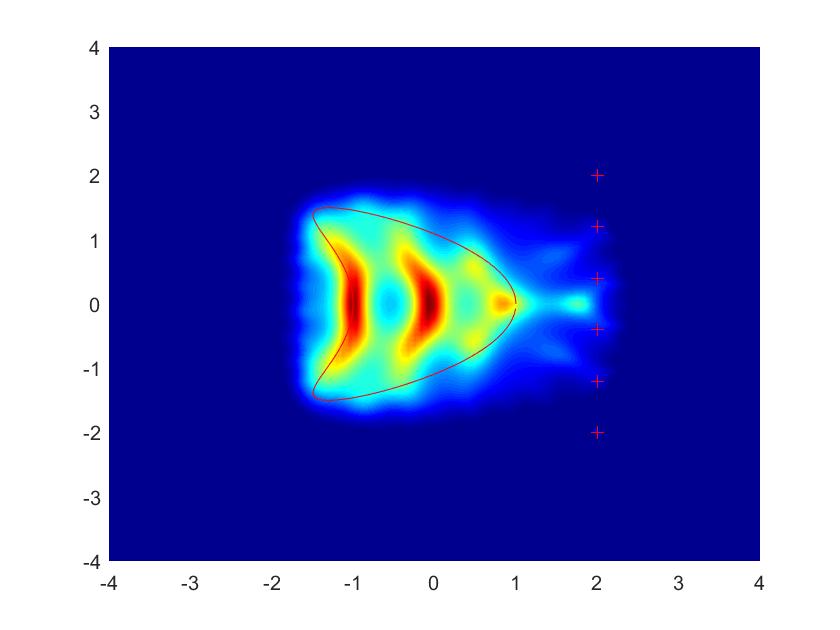} &
\includegraphics[scale=0.2]{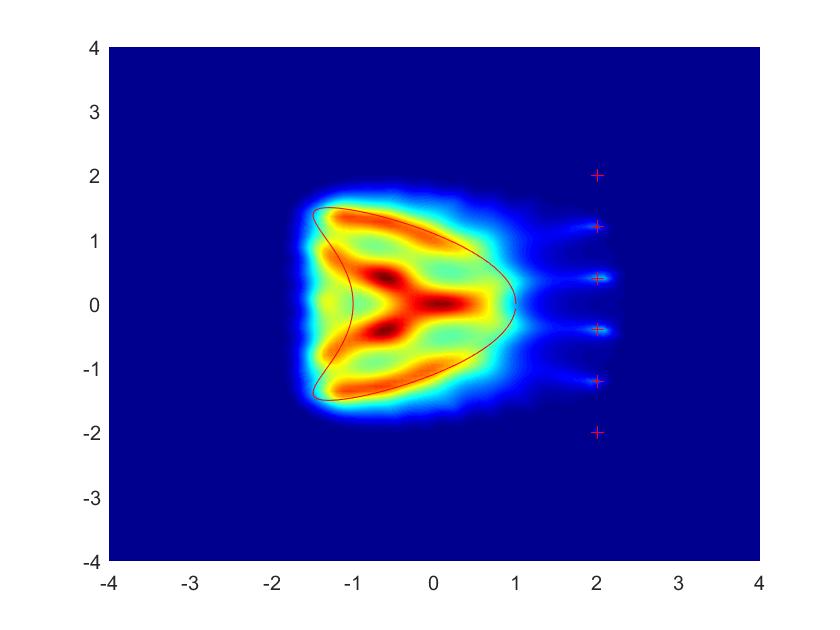}\\
(a) $\mbox{PP}$ & (b) $\mbox{PP}$\\
\includegraphics[scale=0.2]{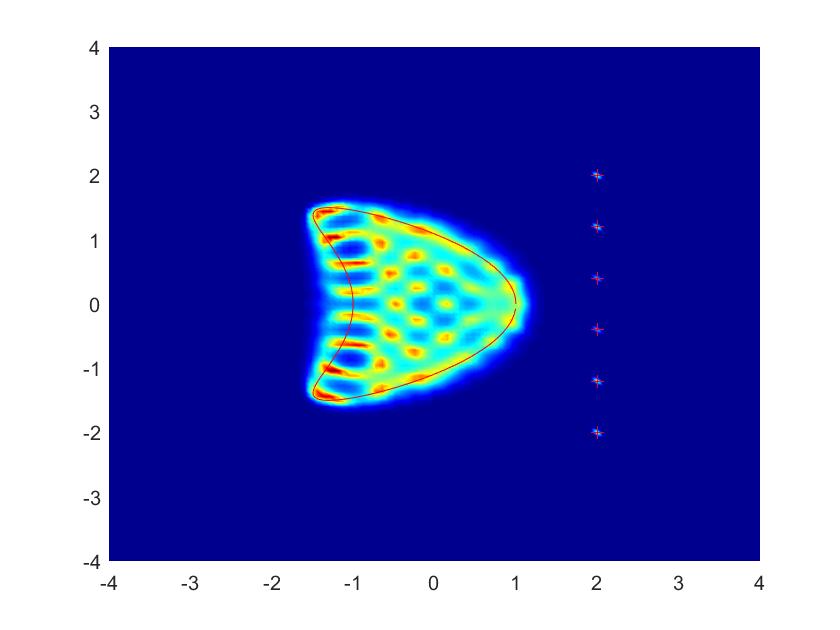} &
\includegraphics[scale=0.2]{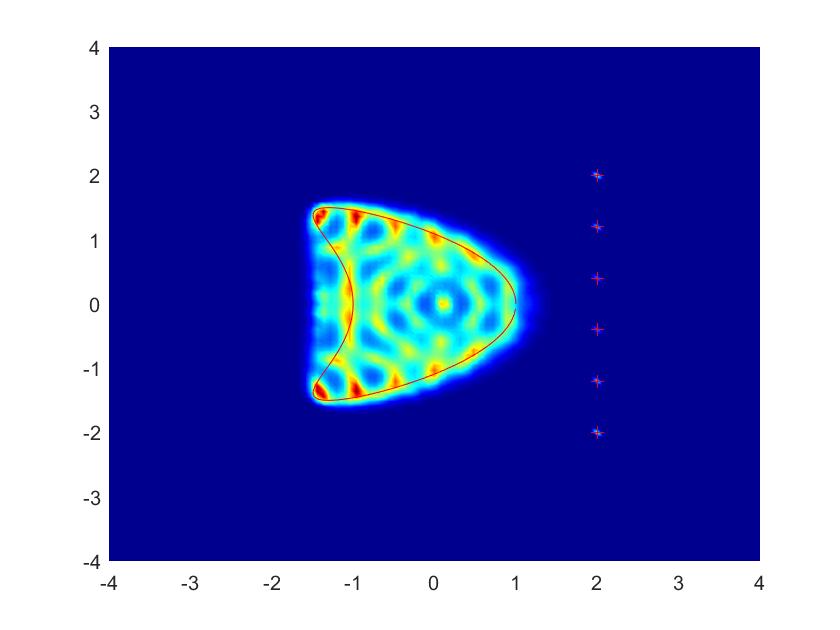}\\
(c) $\mbox{SS}$ & (d) $\mbox{SS}$\\
\includegraphics[scale=0.2]{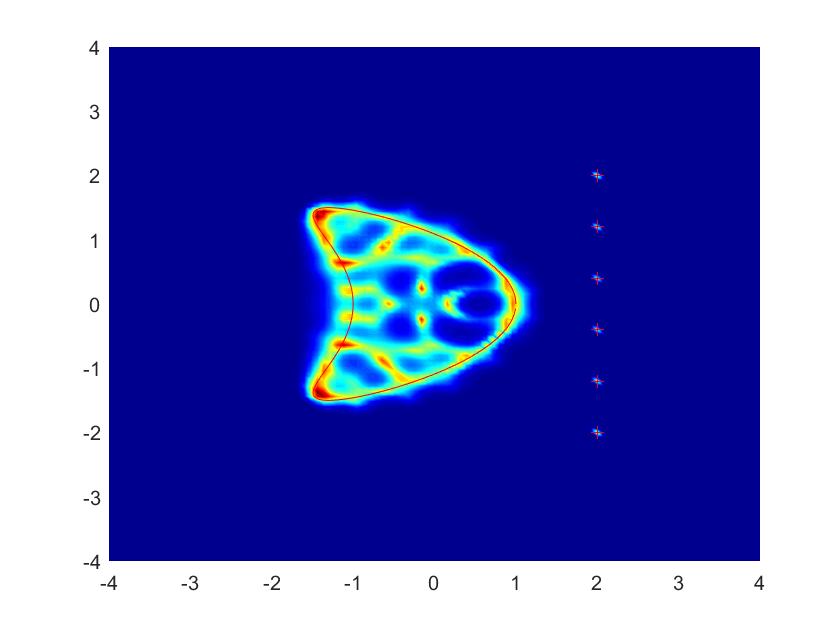} &
\includegraphics[scale=0.2]{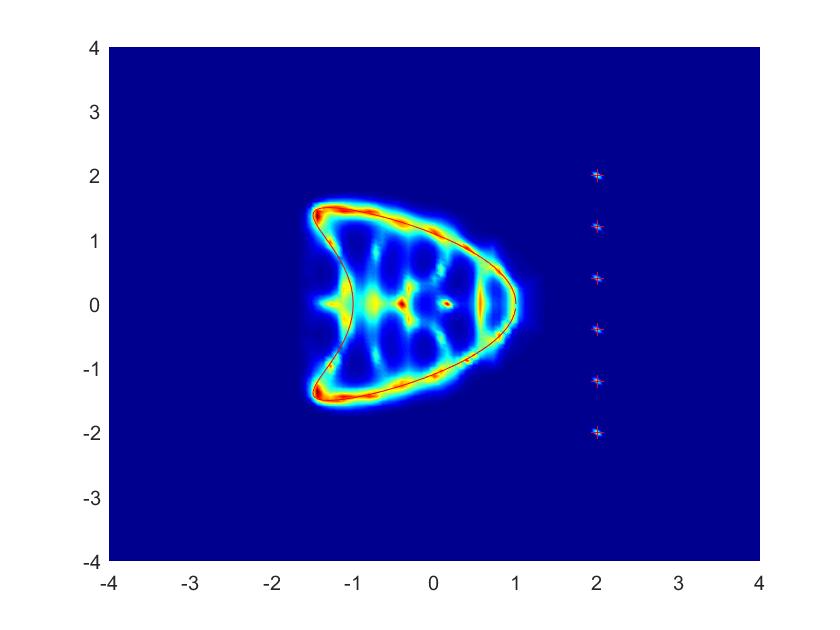}\\
(e) $\mbox{FF}$ & (f) $\mbox{FF}$%
\end{tabular}
\caption{Reconstruction of the kite-shaped obstacle and 6 point-like
scatterers for Example 1 with different polarization vectors $\mathbf{a}%
=(\cos\beta,\sin\beta)$. We set $\beta=0$ in (a,c,e) and $\beta=\pi/2$ in
(b,d,f).}%
\label{E1.1}%
\end{figure}

\begin{figure}[th]
\centering
\begin{tabular}
[c]{cc}%
\includegraphics[scale=0.2]{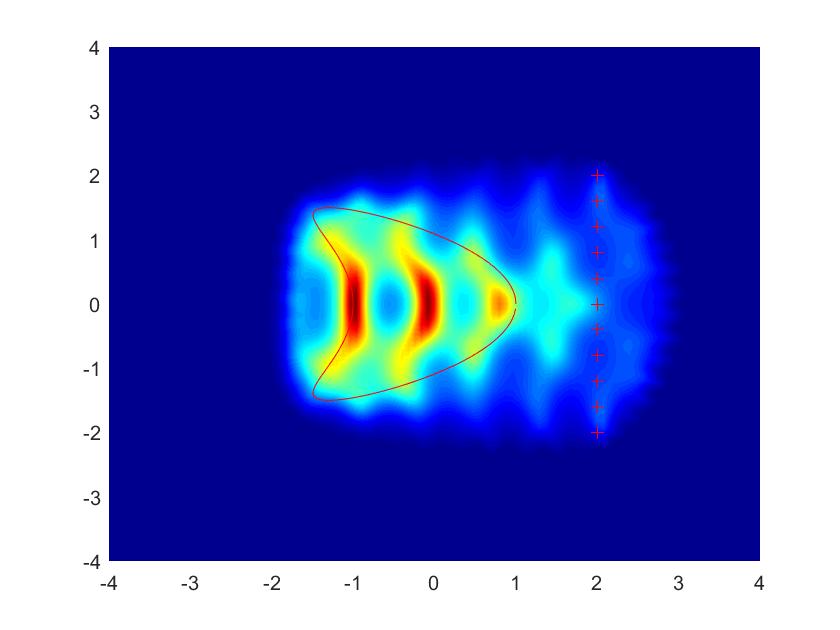} &
\includegraphics[scale=0.2]{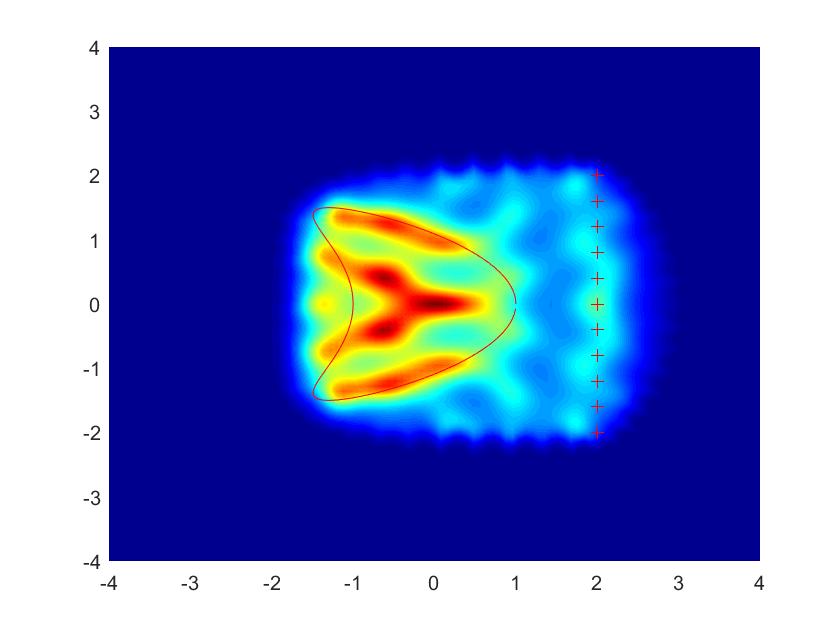}\\
(a) $\mbox{PP}$ & (b) $\mbox{PP}$\\
\includegraphics[scale=0.2]{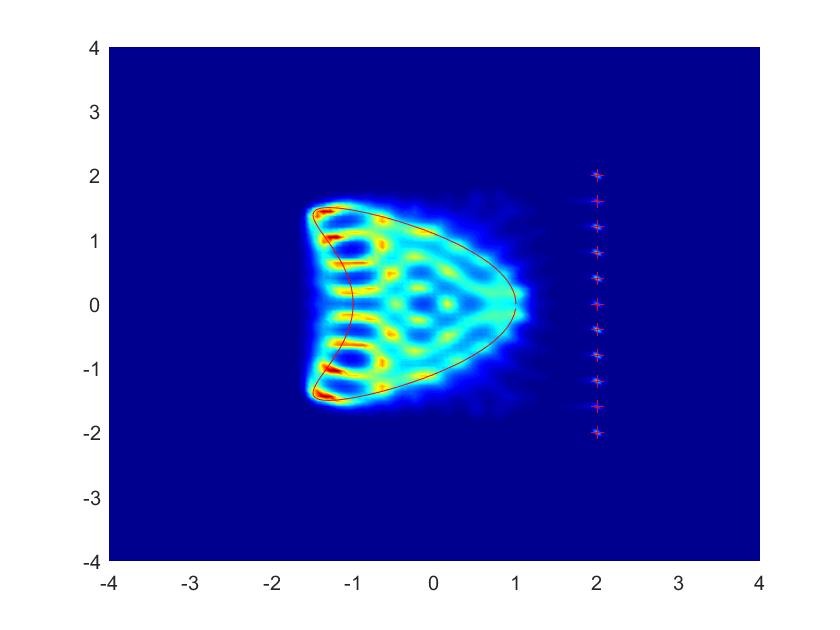} &
\includegraphics[scale=0.2]{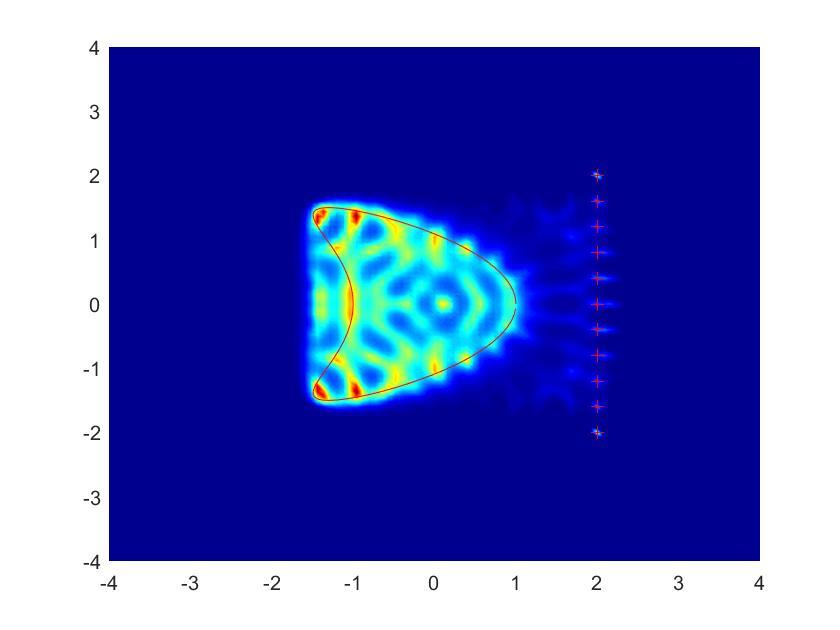}\\
(c) $\mbox{SS}$ & (d) $\mbox{SS}$\\
\includegraphics[scale=0.2]{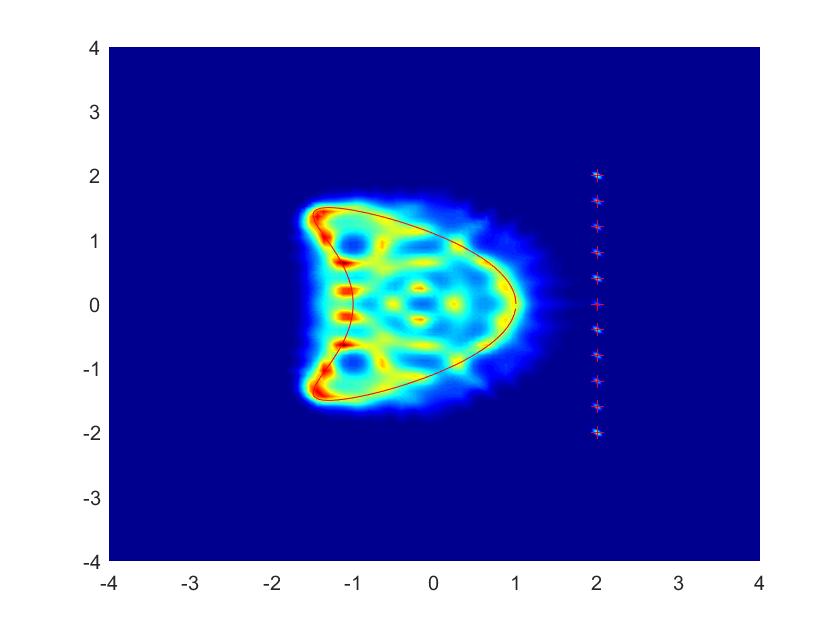} &
\includegraphics[scale=0.2]{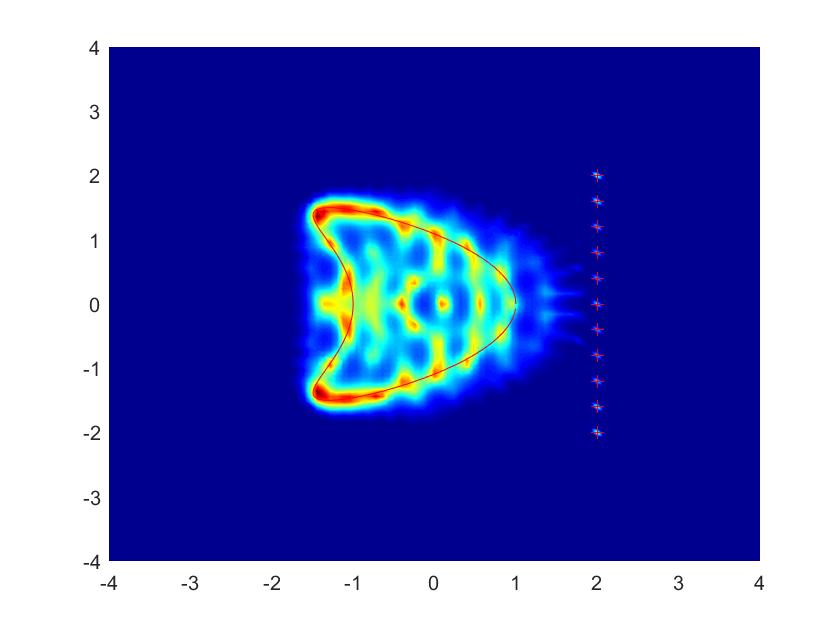}\\
(e) $\mbox{FF}$ & (f) $\mbox{FF}$%
\end{tabular}
\caption{Reconstruction of the kite-shaped obstacle and 11 point-like
scatterers for Example 1 with different polarization vectors $\mathbf{a}%
=(\cos\beta,\sin\beta)$. $\alpha=0$ in (a,c,e), $\beta=\pi/2$ in (b,d,f).}%
\label{E1.2}%
\end{figure}

\begin{figure}[th]
\centering
\begin{tabular}
[c]{ccc}%
\includegraphics[scale=0.15]{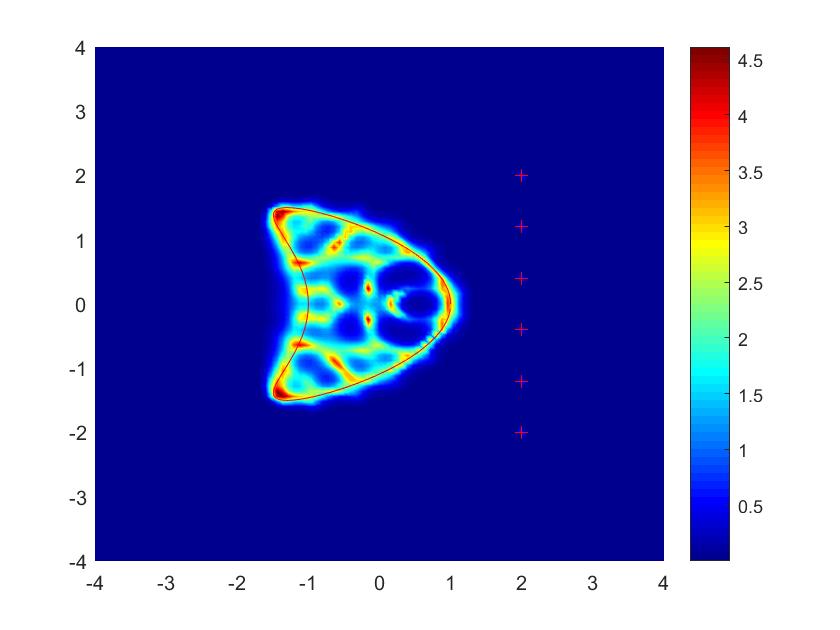} &
\includegraphics[scale=0.15]{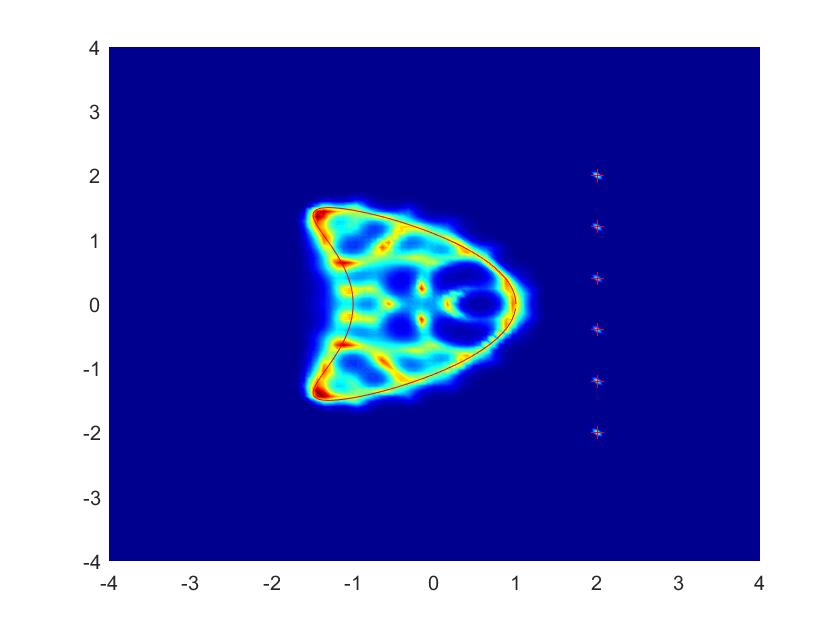} &
\includegraphics[scale=0.15]{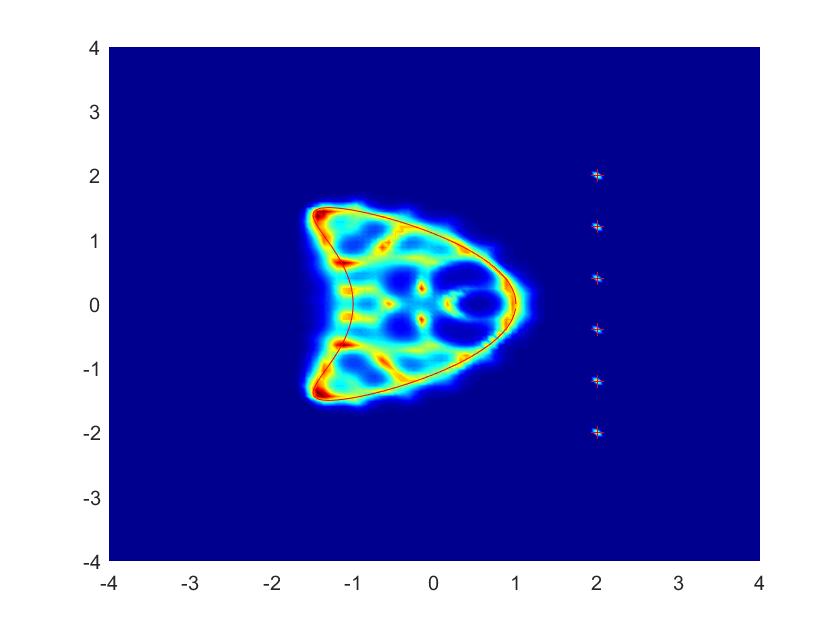}\\
(a) $\alpha_{j}=1$ & (b) $\alpha_{j}=0.1$ & (c) $\alpha_{j}=0.01$\\
\includegraphics[scale=0.15]{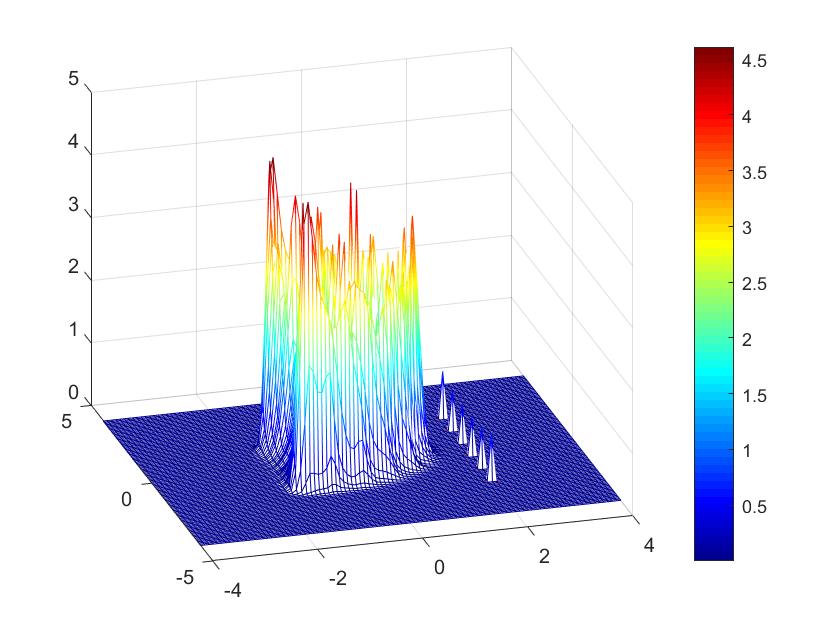} &
\includegraphics[scale=0.15]{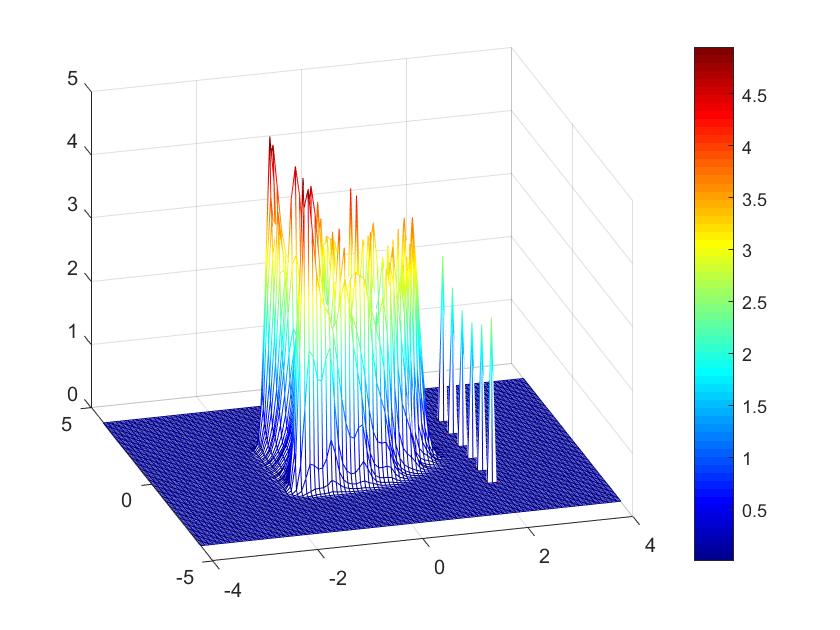} &
\includegraphics[scale=0.15]{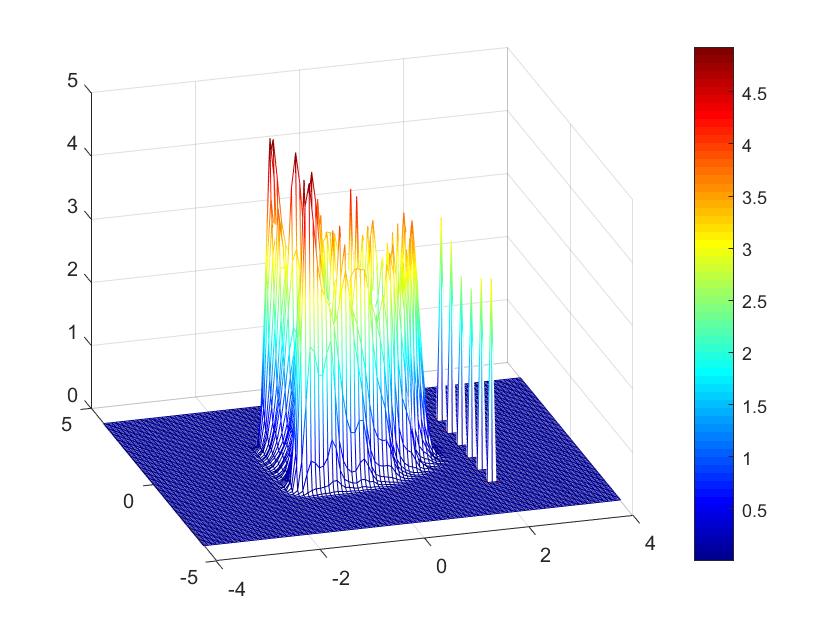}\\
(d) $\alpha_{j}=1$ & (e) $\alpha_{j}=0.1$ & (f) $\alpha_{j}=0.01$\\
&  &
\end{tabular}
\caption{Reconstruction of the kite-shaped obstacle and 6 point-like
scatterers for Example 2 with different ``impedance'' coefficients $\alpha
_{j}, j=1,\cdots,M$..}%
\label{E2.1}%
\end{figure}

\begin{figure}[th]
\centering
\begin{tabular}
[c]{ccc}%
\includegraphics[scale=0.15]{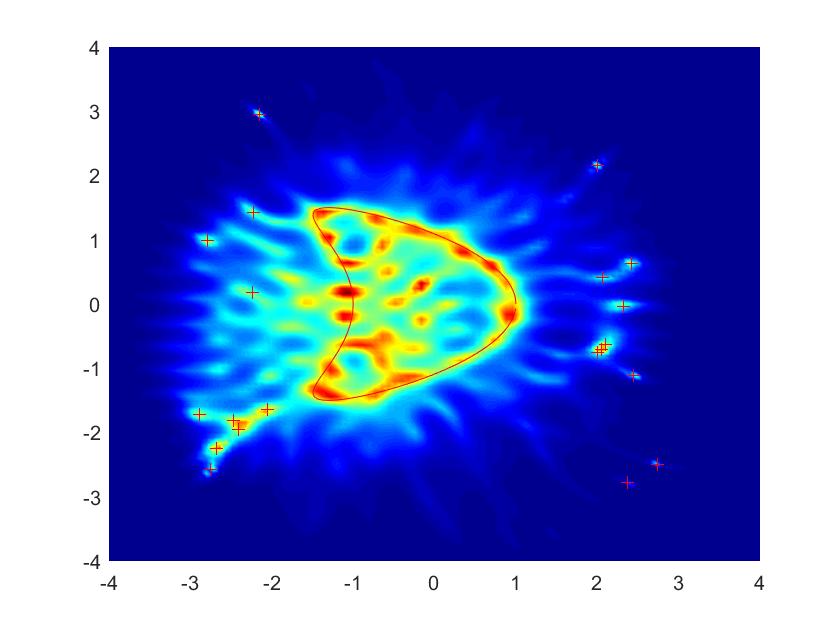} &
\includegraphics[scale=0.15]{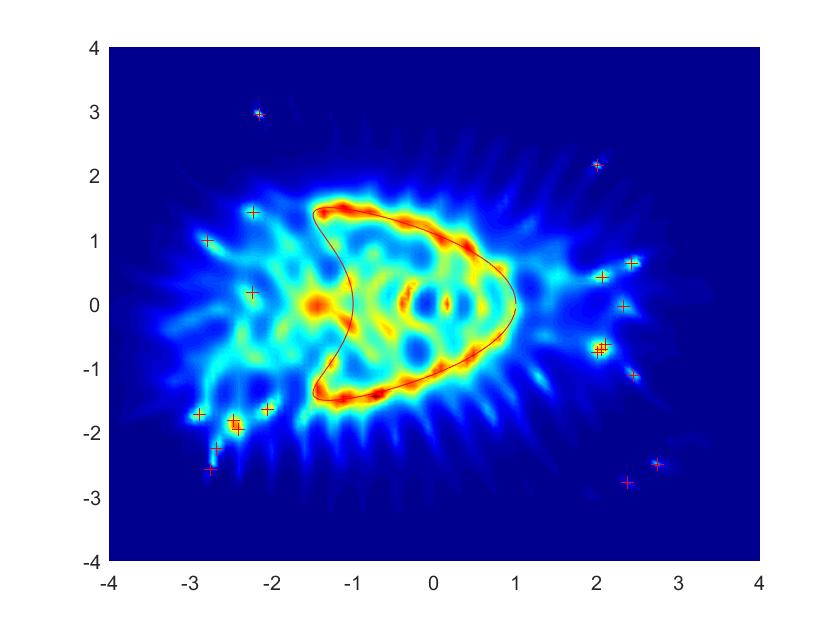} &
\includegraphics[scale=0.15]{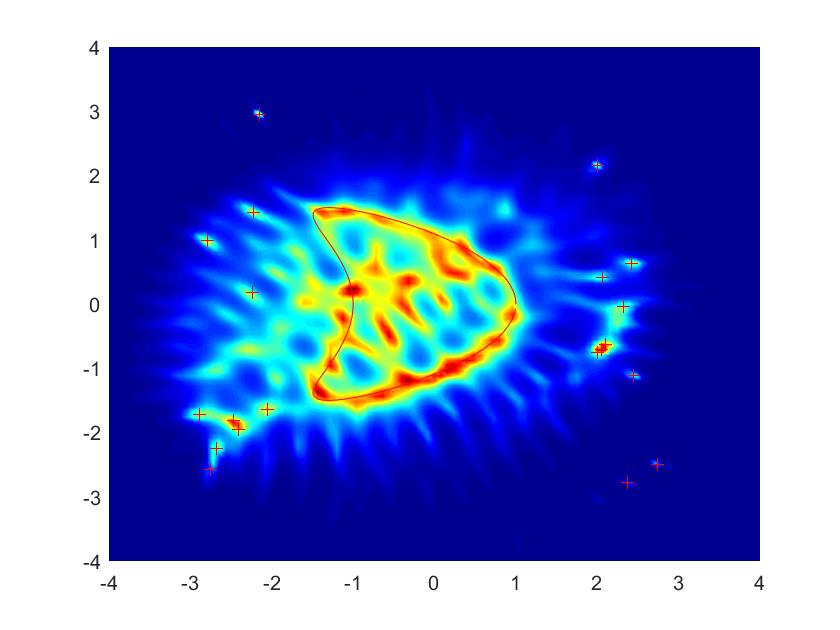}\\
(a) $\beta=0$ & (b) $\beta=\pi/2$ & (c) $\beta=5\pi/4$\\
&  &
\end{tabular}
\caption{Reconstruction of the kite-shaped obstacle and 20 point-like
scatterers for Example 3 with different polarization vectors $\mathbf{a}%
=(\cos\beta,\sin\beta)$.}%
\label{E3.1}%
\end{figure}\FloatBarrier

\section*{Acknowledgements}

Part of this work was finished when G. Hu and A. Mantile visited the Johann Radon Institute
for Computational and Applied Mathematics (RICAM) in October, 2019. The
hospitality of the institute is appreciated.
The authors were partially supported by the Austrian Science Fund (FWF) P28971-N32, \tcr{the NSFC grant 11671028 and NSAF grant U1530401.}

\end{document}